\documentclass[pmlr]{jmlr}

\usepackage{booktabs}
\usepackage[load-configurations=version-1]{siunitx} 

\theorembodyfont{\upshape}
\theoremheaderfont{\scshape}
\theorempostheader{:}
\theoremsep{\newline}

\jmlrvolume{1}
\jmlryear{2020}
\jmlrworkshop{NeurIPS 2020 Competition and Demonstration Track}

\usepackage[subtle]{savetrees} 
\usepackage{arydshln}
\usepackage{wrapfig}
\usepackage{epigraph}
\usepackage{multirow}
\usepackage{algorithm}
\usepackage{algorithmic}

\usepackage{csquotes}
\usepackage{hyperref}       
\usepackage{url}            
\usepackage{booktabs}       
\usepackage{nicefrac}       
\usepackage{microtype}      
\usepackage{color}  
\usepackage{scalerel}
\usepackage[toc,page]{appendix}
\usepackage[font=small,labelfont=bf]{caption}
\usepackage{blindtext}

\usepackage{dsfont}
\usepackage{amsmath}
\usepackage{amsfonts}
\usepackage{amsfonts} 
\usepackage{amssymb}
\usepackage{mathtools}
\usepackage{float}
\newfloat{algorithm}{t}{lop}
\graphicspath{ {images/} }
\usepackage{letltxmacro}%
\usepackage{thmtools,thm-restate}
\usepackage{relsize}
\usepackage{hyperref}
\usepackage{lstautogobble}
\usepackage{float}
\usepackage{adjustbox}
\usepackage{listings,multicol}
\usepackage{xcolor}
\definecolor{codegreen}{rgb}{0,0.6,0}
\definecolor{codegray}{rgb}{0.5,0.5,0.5}
\definecolor{codepurple}{rgb}{0.58,0,0.82}
\definecolor{backcolour}{rgb}{0.95,0.95,0.92}
\lstdefinestyle{mystyle}{
    backgroundcolor=\color{backcolour},   
    commentstyle=\color{codegreen},
    keywordstyle=\color{magenta},
    numberstyle=\tiny\color{codegray},
    stringstyle=\color{codepurple},
    basicstyle=\ttfamily\footnotesize,
    breakatwhitespace=false,         
    breaklines=true,                 
    captionpos=b,                    
    keepspaces=true,                 
    numbers=none,                    
    numbersep=5pt,                  
    showspaces=false,                
    showstringspaces=false,
    showtabs=false,                  
    tabsize=2,
    escapeinside={<@}{@>}
}
\usepackage{footnote}
\makesavenoteenv{tabular}
\makesavenoteenv{table}
\usepackage{balance}  
\usepackage{lipsum}  

\usepackage{enumitem}
\newlist{inparaenum}{enumerate}{2}
\setlist[inparaenum]{nosep}
\setlist[inparaenum,1]{label=\bfseries\arabic*.}
\setlist[inparaenum,2]{label=\bfseries\roman*)}

\usepackage{bm}
\usepackage{breqn}
\usepackage{physics}
\usepackage{enumitem}
\usepackage[colorinlistoftodos]{todonotes}

\usepackage[symbol*]{footmisc}
\DefineFNsymbolsTM{myfnsymbols}{
  \textasteriskcentered *
  \textdagger    \dagger
  \textdaggerdbl \ddagger
  \textsection   \mathsection
  \textbardbl    \|%
  \textparagraph \mathparagraph
}%
\setfnsymbol{myfnsymbols}

\usepackage{multicol}

\usepackage{multicol,caption}
\usepackage{graphicx}
\usepackage{caption}

\newcommand{\bM}{\boldsymbol{M}}

\newcommand{\bU}{\boldsymbol{U}}
\newcommand{\bV}{\boldsymbol{V}}

\newcommand{\bZ}{\boldsymbol{Z}}

\newcommand{\bS}{\boldsymbol{S}}

\newcommand{\btS}{\tilde{\bS}}
\newcommand{\btU}{\tilde{\bU}}
\newcommand{\btV}{\tilde{\bV}}
\newcommand{\ts}{\tilde{s}}

\newcommand{\tS}{\tilde{S}}
\newcommand{\tU}{\tilde{U}}
\newcommand{\tV}{\tilde{V}}

\newcommand{\btZ}{\widetilde{\bZ}}

\newcommand{\bhM}{\widehat{\bM}}

\newcommand{\bhtM}{\hat{\tilde{\bM}}}

\newcommand{\Ex}{\mathbb{E}}

\newcommand{\Rb}{\mathbb{R}}
\newcommand{\Reals}{\mathbb{R}}
\newcommand{\Nb}{\mathbb{N}}
\newcommand{\Zb}{\mathbb{Z}}

\newcommand{\diag}{\text{diag}}

\newcommand{\cI}{{\mathcal I}}

\newcommand{\trow}{t_{\text{row}}}
\newcommand{\tcol}{t_{\text{col}}}

\newtheorem{prop}{Proposition}[section]

\newtoggle{FIGURES}
\togglefalse{FIGURES}
\toggletrue{FIGURES}

\title{tspDB: Time Series Predict DB }

\author{\Name{Anish Agarwal} \Email{anish90@mit.edu}\\
\Name{Abdullah Alomar} \Email{aalomar@mit.edu}\\
\Name{Devavrat Shah} \Email{devavrat@mit.edu}\\
 \addr Massachusetts Institute of Technology}
   
\begin{document}
\maketitle
\begin{abstract}
A major bottleneck of the current Machine Learning (ML) workflow is the time consuming, error prone engineering required to get data from a datastore or a database (DB) to the point an ML algorithm can be applied to it. 
This is further exacerbated since ML algorithms are now trained on large volumes of data, yet we need predictions in real-time, especially in a variety of time-series applications such as finance and real-time control systems. 
Hence, we explore the feasibility of directly integrating prediction functionality on top of a data store or DB. 
Such a system ideally: 
(i)~provides an intuitive prediction query interface which alleviates the unwieldy data engineering; 
(ii)~provides state-of-the-art statistical accuracy while ensuring incremental model update, low model training time 
and low latency for making predictions. 
As the main contribution we explicitly instantiate a proof-of-concept, tspDB\footnote{An open source implementation 
of tspDB is available at \href{https://tspdb.mit.edu}{tspdb.mit.edu}.} which directly integrates with PostgreSQL. 
We rigorously test tspDB's statistical and computational performance against the state-of-the-art time series algorithms, including a Long-Short-Term-Memory (LSTM) neural network and DeepAR (industry standard deep learning library by Amazon). 
Statistically, on standard time series benchmarks, tspDB outperforms LSTM and DeepAR with 1.1-1.3x higher relative accuracy. 
Computationally, tspDB is 59-62x and 94-95x faster compared to LSTM and DeepAR in terms of median ML model training time and prediction query latency, respectively. 
Further, compared to PostgreSQL's bulk insert time and its SELECT query latency, tspDB is
slower only by 1.3x and 2.6x respectively. 
That is, tspDB is a real-time prediction system in that its model training / prediction query time is similar to just inserting / reading data from a DB.
As an algorithmic contribution, we introduce an incremental multivariate matrix factorization based time series method, which tspDB is built off.
We show this method also allows one to produce reliable prediction intervals by accurately estimating the time-varying variance of a time series, thereby addressing an important problem in time series analysis.

\end{abstract}

\newpage
\tableofcontents
\newpage


\section{Introduction}\label{sec:introduction}

\textbf{Data Engineering: Major Bottleneck in Modern ML Workflow.}
An important goal in the Systems for ML community has been to make ML more broadly accessible \cite{MLSys}.
Arguably, the major bottleneck towards this goal is not the lack of access to prediction algorithms, for which many excellent open-source ML libraries exist. 
Rather, it is the complex engineering and data processing required to take data from a datastore or database (DB) into a particular work environment format (e.g. spark data-frame) so that a prediction algorithm can be trained, and to do so in a scalable manner.
See Figure \ref{fig:proposed_system} for a visual depiction of the ML workflow as it stands; we aim to alleviate much of this `unwieldy' data engineering by building tspDB, a system that directly integrates prediction functionality on top of a DB in real-time. 

\iftoggle{FIGURES}{
\begin{figure}[htb]
	\centering
	\includegraphics[width=0.9\linewidth]{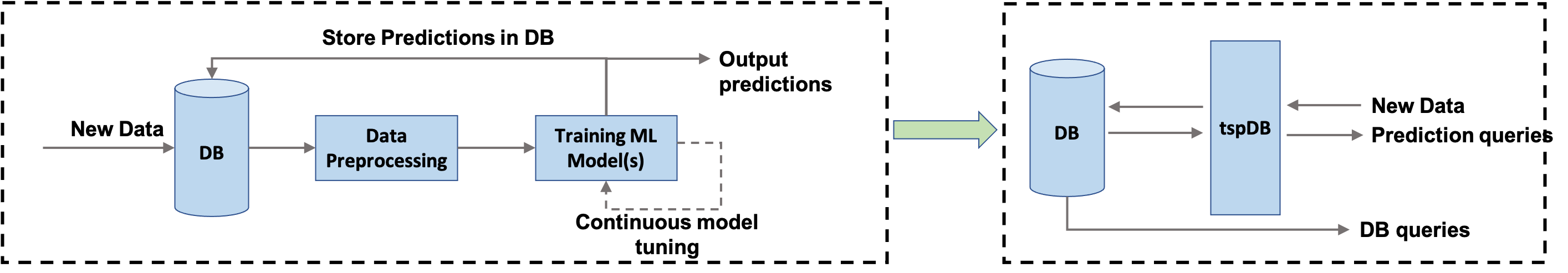}
	\caption{Pictorial depiction of data engineering---data transformation, feature extraction, model training/tuning---as a major bottleneck of the ML workflow. tspDB aims to alleviate it by direct integration of prediction functionality on top of a DB.}
	\label{fig:proposed_system}
\vspace{-3mm} 
\end{figure}

}

\medskip
\noindent
\textbf{The Goal.}  
Towards easing this bottleneck, and increasing accessibility of ML, our objective is to explore the feasibility of directly integrating a prediction system on top of the DB layer. 
This is in line with recent efforts to ``democratize" ML \citep{Tupaq, alpine}. 
Indeed, the authors in \cite{thecase} make a compelling case of the potential gains to be had by direct integration of predictive functionality on top of DBs---by drawing an analogy with the now scalable, robust and mature data management capabilities of DBs, they argue that a similar approach to the management of predictive models can lead to large improvements in both computational performance and accessibility of ML.

\hspace{4mm} In this work, we focus on multivariate time series data, that is, data where each unit (e.g. collection of stocks, temperature readings from a collection of sensors) is associated with a time stamp.
We do so as this is one of the primary ways data is structured in a range of important applications including finance, healthcare, retail, to name a few.
 We consider two types of prediction tasks for time series data:
(i) imputing a missing or noisy observation for a time series datapoint we do observe;
(ii) forecasting a time series data point in the future (i.e., data that is as of yet unobserved).

\medskip
\noindent
{\bf Statistical and Computational Performance Metrics to Benchmark ML Algorithms.}
As described in detail in \cite{MLSys}, given the growing demand to embed ML functionality in high-performance systems, especially in applications with time series data (e.g. financial systems, control systems), it is increasingly vital that we evaluate ML algorithms/systems not just through prediction accuracy, but through computational metrics as well. 
We consider two important computational metrics: (i) the time it takes to train a ML model; (ii) the time it takes to answer a prediction query, i.e., prediction query latency.
Throughout the paper, along with statistical accuracy, we will benchmark computational performance through these two computational metrics.

\begin{table}[tb]
\floatconts
{table:full}
{\caption{
Summary of statistical accuracy (Table \ref{tab:table1_b}), computational performance (Table \ref{tab:table1_a}), and functionalities (Table \ref{tab:table1_c}) of tspDB vs. other state-of-the-art time series algorithms. 
We use a weighted variant of the standard Borda Count (WBC) to evaluate statistical performance (precise definition in Appendix \ref{sec:score_metric}). WBC takes value in $[0,1]$: $0.5$ indicates equal performance to all other algorithms on average, $1$ (resp. $0$) indicates it ``infinitely'' superior (resp. inferior) performance compared to all others. 
}}
{
\vspace{-5mm} 
 \subtable[Summary of statistical accuracy results (see Section \ref{sec:experiments_stat} for details).][t]{%
 \label{tab:table1_b}%
\tabcolsep=0.305cm
\fontsize{8.5pt}{8.9pt}\selectfont
\centering
\begin{tabular}{@{}llcccccll@{}}
\toprule
                                                                           \multirow{4}{*}{\begin{tabular}[c]{@{}l@{}}\textbf{Statistical} \\ \textbf{Accuracy}  \\  \textbf{(WBC)}\end{tabular}}                 &                 & \multicolumn{5}{c}{\textbf{Forecasting}   }                                                  & \multicolumn{2}{c}{\textbf{Imputation}} \\ 
&  & \textbf{tspDB}  & \textbf{LSTM} & \textbf{DeepAR} & \textbf{TRMF} & \textbf{Prophet} & \textbf{tspDB}  & \textbf{TRMF} \\\cmidrule[0.5pt](lr){2-2} \cmidrule[0.5pt](r){3-7} \cmidrule[0.5pt](l){8-9}
                                                                                           & {Mean}            & \textbf{0.597} & 0.446         & 0.559           & 0.541         & 0.358            & \textbf{0.572} & 0.428         \\
                                                                                           & {Variance}        & \textbf{0.726} & N/A           & 0.274           & N/A           & N/A              & \textbf{0.597} & 0.403  \footnote{TRMF does not natively support variance imputation. We use an adapted version of TRMF since no method in the literature exists to do variance imputation to the best of our knowledge. Refer to Appendix \ref{sec:exp_appendix_methods} for details.}        \\ \bottomrule \\
\end{tabular}
}


 \subtable[Summary of computational  performance results (see Section \ref{sec:computational_experiments} for details).][t]{
 \label{tab:table1_a}
\tabcolsep=0.10cm
\fontsize{8.5pt}{8.9pt}\selectfont
\centering
\begin{tabular}{@{}cccccccc@{}}
\toprule
\multicolumn{1}{c}{}                       \multirow{3}{*}{\begin{tabular}[c]{@{}c@{}}\textbf{Computational} \\  \textbf{Performance}\end{tabular}}                                             & \textbf{Metric}                                                                      & \textbf{tspDB} & \textbf{LSTM} & \textbf{DeepAR} & \textbf{TRMF} & \textbf{Prophet} & \textbf{PostgreSQL} \\ \cmidrule[0.5pt](l){2-8}
& \begin{tabular}[c]{@{}c@{}}{Training/Insert  Time (seconds)}\end{tabular} & \textbf{9.00} & 556.1         & 531.7           & 13.6          & 3445.9           & 6.64                \\
                                                                                       & \begin{tabular}[c]{@{}c@{}}{Prediction/Querying  Time  (ms)}\end{tabular} & \textbf{3.45} & 325.2         & 326.2           & 5.28          & 1466.1           & 1.32                \\ \bottomrule  \\
\end{tabular}
}

 \subtable[Summary of tspDB functionalities compared to other time series forecasting algorithms.][t]{
  \label{tab:table1_c}
\tabcolsep=0.3cm
\fontsize{8.5pt}{8.9pt}\selectfont
\centering
\begin{tabular}{@{}llccccc@{}}
\toprule
                                 &                               & \textbf{tspDB} & \textbf{LSTM} & \textbf{DeepAR} & \textbf{TRMF} & \textbf{Prophet} \\ \cmidrule[0.5pt](l){2-7}
\multirow{3}{*}{\textbf{Functionalities}} & {Multivariate Time Series}    & \textbf{Yes}  & \textbf{Yes}  & \textbf{Yes}    & \textbf{Yes}  & No               \\
                                 & {Variance Estimation}         & \textbf{Yes}  & No            & \textbf{Yes}    & No            & No               \\
                                 & {Working with Missing Values} & \textbf{Yes}  & No            & \textbf{Yes}    & \textbf{Yes}  & \textbf{Yes}     \\ \bottomrule \\
\end{tabular}
\label{tab:table1_c}
}

}
\vspace{-5mm} 
\end{table}

%

\subsection{Contributions}
\textbf{tspDB.} 
The main contribution of this work is tspDB, a real-time prediction system for time series data, which aims to alleviate the data engineering bottleneck by providing direct access to predictive functionality for time series data (see Section \ref{sec:interface} for the prediction interface tspDB provides to users). 
We find tspDB outperforms the various popular existing time series libraries we compare against with respect to all three metrics we consider---statistical performance, speed of training the ML model and latency in making a prediction. 
See Table \ref{table:full} for a summary of the performance of tspDB, and the other time series algorithms we compare against, along these metrics.
Details of the precise experiments run to produce Table \ref{table:full} can be found in Section \ref{sec:experiments}.

%
{\bf Novel Multivariate Time Series Algorithm to Estimate Mean and Variance.}
tspDB is built of a novel time series prediction algorithm that we propose.
It is a generalization of the method proposed in \cite{sigmetrics} and extends it in the following two ways:
(i) it provides predictions for multivariate time series data, which performs significantly better than the original univariate algorithm (see Table \ref{table:multi_vs_uni});
(ii) it allows one to estimate the time-varying variance of a multivariate time series (a la GARCH-like models), thereby addressing an important issue in time series analysis of how to estimate the volatility of the time series.
We note this algorithm can be viewed as a variant of multivariate Singular Spectrum Analysis (mSSA).
Refer to Appendix \ref{sec:related_work} for detailed comparison with \cite{sigmetrics} and mSSA.
Details of the algorithm can be found in Section \ref{sec:overview}.

\hspace{4mm} 
An important feature of the proposed mean and variance estimation algorithm is that it is ``noise agnostic", i.e., it effectively imputes and forecasts both the time-varying mean and variance {\em regardless} of the noise model---e.g. Gaussian noise (continuous observations), Poisson and Binomial noise (integer observations), Bernoulli noise (binary observations).
tspDB's ability to accurately estimate the time-varying mean and variance of a time series and its robustness to different noise models leads to interesting applications.
For example, it allows one to verify whether or not a set of integer observations are generated from a Poisson process in a data-driven way as shown in Figure \ref{fig:poisson_vs_binomial}, a task that has received significant attention in the literature (see \cite{poisson1, poisson2, poisson3, poisson4}). 
See Section \ref{sec:exp_noise} for further empirical evidence of tspDB's robustness to noise.

\vspace{1mm}
{\em Theoretical Justification.} 
Given our focus on actually building a real-time time series prediction system and comprehensively testing its statistical and computational performance, a rigorous theoretical analysis of the proposed algorithm is beyond the scope of this paper.
%
%
However, in Section \ref{sec:justification} we provide intuitive formal justification of when and why this algorithm works.



\vspace{1mm}
{\em Computationally Efficient, Incremental Implementation of Algorithm in tspDB.} 
To ensure that tspDB achieves high computational performance (i.e., has low model training and prediction query time) without sacrificing statistical accuracy, we develop and implement a scalable, incremental variant of the  algorithm we propose in Section \ref{sec:overview}.
Details of this computationally efficient variant can be found in Section \ref{sec:inc_algorithm}.
To understand the statistical and computational tradeoffs in tspDB, we run extensive experiments on how the three metrics we consider tradeoff as we vary key hyper-parameters in our implementation of tspDB (details can be found in Section \ref{sec:appendix_stat_experiments}).

\begin{figure}[htbp]
\floatconts
{fig:poisson_vs_binomial}
{\caption{Does the time series follow a Poisson process? With no prior knowledge of the distribution, tspDB can be used to verify if integer observations are generated from a Poisson process (Figure \ref{fig:poisson_var}) by checking if the latent mean and variances are equal, or if they are different (e.g. Binomial in Figure \ref{fig:binomial_var}).}}
{
	\subfigure[][t]{\label{fig:poisson_var}
	\includegraphics[width=0.45\linewidth]{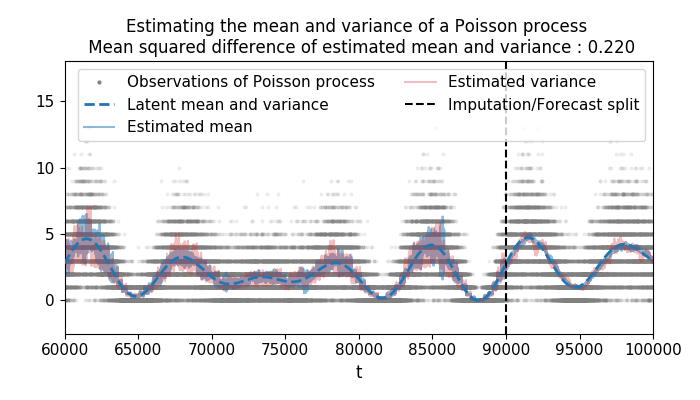}}
	 \qquad  	  \subfigure[][t]{\label{fig:binomial_var}
	\includegraphics[width=0.45\linewidth]{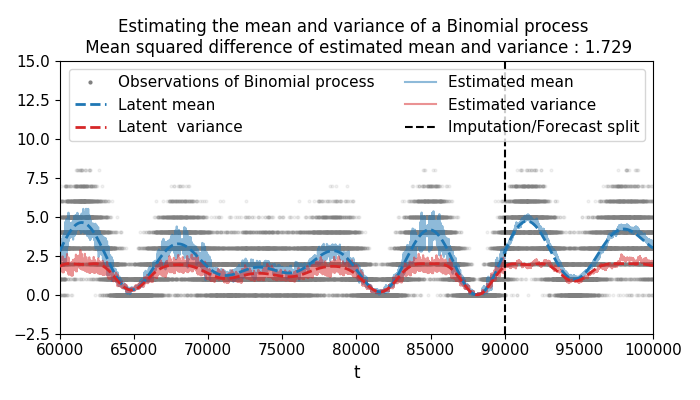}}
	\vspace{-5mm} 
}
\vspace{-2mm} 
\end{figure}

{\bf Extensive Statistical and Computational Benchmarking of tspDB.}
In Section \ref{sec:experiments}, we extensively benchmark tspDB's statistical and computational performance against popular state-of-the-art prediction libraries, and just its computational performance against PostgreSQL---results of these experiments are summarized in Table \ref{table:full}. 

\vspace{1mm}
{\em Superior Statistical Accuracy.} 
In Section \ref{sec:experiments_stat}, we verify tspDB's statistical accuracy by testing its performance on benchmark datasets of time series data from real-world applications (e.g. used in \cite{TRMF}) and also on synthetic datasets we generate. 
We compare tspDB's accuracy against state-of-the-art algorithms, including LSTMs \citep{LSTM}, DeepAR (by Amazon \cite{DeepAR}), Temporal Regularized 
Matrix Factorization (TRMF) \citep{TRMF} and Prophet  (by Facebook  \cite{Prophet}).
We measure the imputation and forecasting prediction accuracy, for both mean and variance, of these various algorithms as we vary the level of 
noise, increase the fraction of missing data, and change the noise model. 
We use the `Weighted Borda Count' (WBC) as our statistical accuracy metric.
This score is based on pairwise comparisons of the normalized root mean squared error (NRMSE) across experiments.
We use it as it better captures the relative performance across all experiments and methods, unlike a simple statistic such as the mean or the median 
\footnote{In fact, as evidenced in Table \ref{table:stat}, tspDB has even better statistical performance if we use the mean NRMSE score, as is standard practice.}.
Refer to Appendix \ref{sec:score_metric} for the definition of WBC and NRMSE. 

\hspace{4mm} On standard time series benchmarks, using WBC as our metric, we find tspDB outperforms all other methods in both mean imputation and forecasting.
We highlight tspDB's mean WBC score is 1.1-1.3x higher compared to LSTM and DeepAR, some of the most commonly used modern deep learning based time series algorithms. 
For variance estimation, which we use to produce prediction intervals as described in Section \ref{sec:system_query}, we again find tspDB outperforms all other methods.
Importantly, we note that existing methods do not provide variance estimation and instead we have to adapt these other algorithms to be able to make such comparisons---we adapt TRMF to do variance imputation and use DeepAR's in-built uncertainty quantification functionality for variance forecasting.
Despite providing such an `unfair' advantage to these existing approaches, tspDB still manages to outperform them.
In short, tspDB achieves state-of-the-art statistical performance.
%
%
%

\vspace{1mm}
{\em Real-time model training, low-latency predictions.}
As stated earlier, the two metrics we use to measure computational performance of tspDB are ML model training time and prediction query latency. 
In Section \ref{sec:computational_experiments}, we test how tspDB compares against the popular, open-source time series prediction libraries stated above---LSTM, DeepAR, TRMF, and Prophet.
With respect to both ML model training time and prediction query latency, tspDB outperforms all other prediction algorithms we compare against---refer to Section \ref{sec:computational_experiments} for a precise description of how ML model training time and prediction latency are measured.
We highlight that compared to LSTM and DeepAR, tspDB's median ML model training time and prediction latency is $59$-$62$x and $94$-$95$x quicker, respectively. 
See Table \ref{table:full} for a summary of the computational performance of tspDB compared to the other methods.
In conclusion, tspDB achieves both state-of-the-art statistical and computational performance.
\vspace{-2mm}

\vspace{-2mm}
\section{Prediction Interface}\label{sec:interface}
\vspace{-2mm}
An important goal of this work is to design an `easy-to-use' interface to make predictions, which: 
(i) alleviates the need for error-prone data engineering; 
(ii) directly gives access to predictive functionality in real-time for time series data. 
In particular, we wish to provide access to imputation and forecasting functionality for multivariate time series data.  
Towards this goal, there has been considerable recent work, especially in industry, exploring the feasibility of DBs to automate 
and natively support ML workloads (\cite{cxoracle, pymssql, ibmdb, RODM, RevoScaleR, idmdbR}). 
The focus of these works has been to expose an interface that allows users to select from an array of 
ML algorithms (e.g. generalized linear models, random forests, neural networks) and train them (plus tune hyper-parameters) in the DB itself. 

\vspace{2mm}
\textbf{Prediction Query.} In contrast, a significant point of departure of tspDB is how an end user interfaces with the system to produce 
predictions. In particular, to make ML more broadly accessible, we take a different approach and abstract  the ML 
model from the user, and instead \textit{strive for a single interface to answer both standard DB queries and 
predictive queries}. In particular, in tspDB, a predictive query has the same form as a standard SELECT query. 
The key difference is that in a predictive query, the response is a \textit{prediction}, rather than a retrieval 
of available data.  
For example, consider a relation of the stock price of three companies over 100 days: \texttt{stock(day:timestamp, company1:float, company2:float, company3:float)}. 
An example of a predictive query in this setting is shown in Figure \ref{fig:predict_query}. 
Note, day 101 is a forecast, i.e., a prediction, since the DB only contains data for the first 100 days.
Similarly, if the query is changed to \texttt{WHERE day = 10}, then we get the imputed (or de-noised) value for the stock price on the tenth day.
In this case, the PREDICT query response differs from a standard SELECT query as rather than simply retrieving the available data for that day, which may possibly even be missing, a predicted de-noised value is returned. 

\vspace{2mm}
\textbf{Building a Prediction Model in DB.}
To enable PREDICT queries, we need to build a prediction model using the available multivariate time series data. 
Continuing the example from above, a user can build a prediction model in tspDB as shown in Figure \ref{fig:create_model} using the CREATE query in tspDB.
Note, importantly, the prediction model can be trained over multiple time series columns, i.e., is built by simultaneously using data from multiple time series, in this case the stock prices for the three companies.
%
Thus, the CREATE and PREDICT queries in tspDB are very much like building a DB index and making a SELECT query in SQL. 

\begin{figure}[htbp]
\vspace{-3mm}
\floatconts
{fig:interface}
{\caption{Proposed interface for a PREDICT query (Figure \ref{fig:predict_query}), and CREATE prediction model in tspDB (Figure \ref{fig:create_model}).}}
{
	\subfigure[][c]{\label{fig:predict_query}
	\includeteximage[]{figures/sql_query}}
	\qquad 
	\subfigure[][c]{\label{fig:create_model}
	\includeteximage[]{figures/sql_query_2}}
	\vspace{-6mm} 
}
\vspace{-3mm}
\end{figure}
%
%
\textbf{tspDB vis-a-vis PostgreSQL DB.}
As a further computational test, in Section \ref{sec:postgres_computational_experiments}, we compare the computational performance of tspDB versus the standard DB index of PostgreSQL.
On standard multivariate time series datasets, we find that the time required to CREATE prediction model in tspDB ranges from 0.58x-1.52x to that of PostgreSQL's bulk insert time.
With respect to query latency, we find that the time taken to answer a PREDICT query in tspDB is $1.6$ to $2.8$ times higher compared to answering a standard SELECT query, ($1.6$-$2.7$x for imputation queries, $1.7$-$2.8$x for forecasting queries).
In absolute terms, this translates to about $1.32$ milliseconds to answer a SELECT query and ~$3.36$/$3.45$ milliseconds to answer an imputation/forecasting PREDICT query (see Appendix \ref{sec:system_config} for the machine configuration used).
See Table \ref{table:full} for summary of results.
Hence, tspDB's computational performance is close to the time it takes to just insert and read data from PostgreSQL, making it a real-time prediction system. 
%

%

%


\vspace{2mm}
\textbf{tspDB: Open-Source Implementation.} 
As stated earlier, as our main contribution we explicitly instantiate a proof-of-concept, tspDB, which directly integrates with PostgreSQL, and supports the PREDICT query (and CREATE prediction_model) functionality.  
tspDB has  an open-source implementation and is an extension of PostgreSQL; it allows users to create prediction queries over both single columns or multiple columns on a time series relation, i.e., allows for both univariate and multivariate time series prediction, and provides prediction intervals for the estimates. 
\textit{Importantly, we note that all results presented (e.g., in Table \ref{table:full}) can be reproduced using our open-source implementation with default settings}
\vspace{-3mm}

\vspace{-1mm}
\section{Incremental Matrix Factorization Based Time Series Algorithm}\label{sec:algorithm}
\vspace{-2mm}

In this section, we describe our novel incremental algorithm for multivariate time series prediction. 
In Section \ref{sec:overview},  we describe the batch version of the algorithm used to do mean and variance estimation;
In Section \ref{sec:inc_algorithm} we describe the incremental variant of the algorithm that is tspDB is actually built off.
In Section \ref{sec:justification}, we provide intuitive formal justification for this algorithm. 
In Section \ref{sec:tspDB_implementation_details}, we detail the system implementation of tspDB's direct integration with PostgreSQL.

\vspace{-2mm}
\subsection{Algorithm Description}\label{sec:overview}

\vspace{1mm}
\textbf{Multivariate Time Series Notation.}
We begin with some brief notation necessary to describe the algorithm.
Without loss of generality, we consider a discrete time index $t \in \Zb$
\footnote{This is not restrictive as tspDB supports prediction tasks on time series observed at different frequencies. It does so by allowing the user to pick an aggregation interval (e.g. second, minute, hour) and an aggregation function (e.g. mean, min, max).}. 
Say we aim to train a prediction model using $N$ time series, and for each time series $n \in [N] := \{1, \dots, N\}$, we have access to $T \in \Nb$ observations (the observations might be missing).
Let $X_n(t)$ denote the (potentially noisy, missing) observation of the $n$-th time series, $X_n$, at the $t$-th time step.
For any $t, s \in [T]$ such that $t<s$, let $X_n(t:s) := [X_n(t), \dots, X_n(s)]$.  
Let $X$ denote the collection of $N$ time series $X_1, \dots, X_N$.

\vspace{1mm}
\textbf{Page Matrix Data Representation.} 
The key data representation that is utilized for the proposed algorithm is the ``stacked'' Page matrix.
Let integers $L, P \geq 1$ be such that $P = \lfloor T/L \rfloor$. Let $\bar{P}:=N\times P$.
Then the stacked Page matrix, $\bZ^{X} \in \Rb^{L \times \bar{P}}$ induced by $X$ is defined as  
\[
Z^{X}_{i,[j+P\times(n-1)N]} = X_n(i+(j-1)L),~~\mbox{for}~~i\in[L], j\in[P], n\in[N].
\]
Further, let $\bZ^{X^2} \in \Rb^{L \times \bar{P}}$ denote the stacked Page matrix of squared observations,
\[
Z^{X^2}_{i,[j+P\times(n-1)N]} = X^2_n(i+(j-1)L).
\] 

\textbf{Mean Estimation Algorithm.} 
Figure \ref{fig:algorithm} provides a visual depiction of the core steps of the algorithm. \\
\iftoggle{FIGURES}{
\begin{figure}[!htb]
	\centering
	\includegraphics[width=0.8\linewidth]{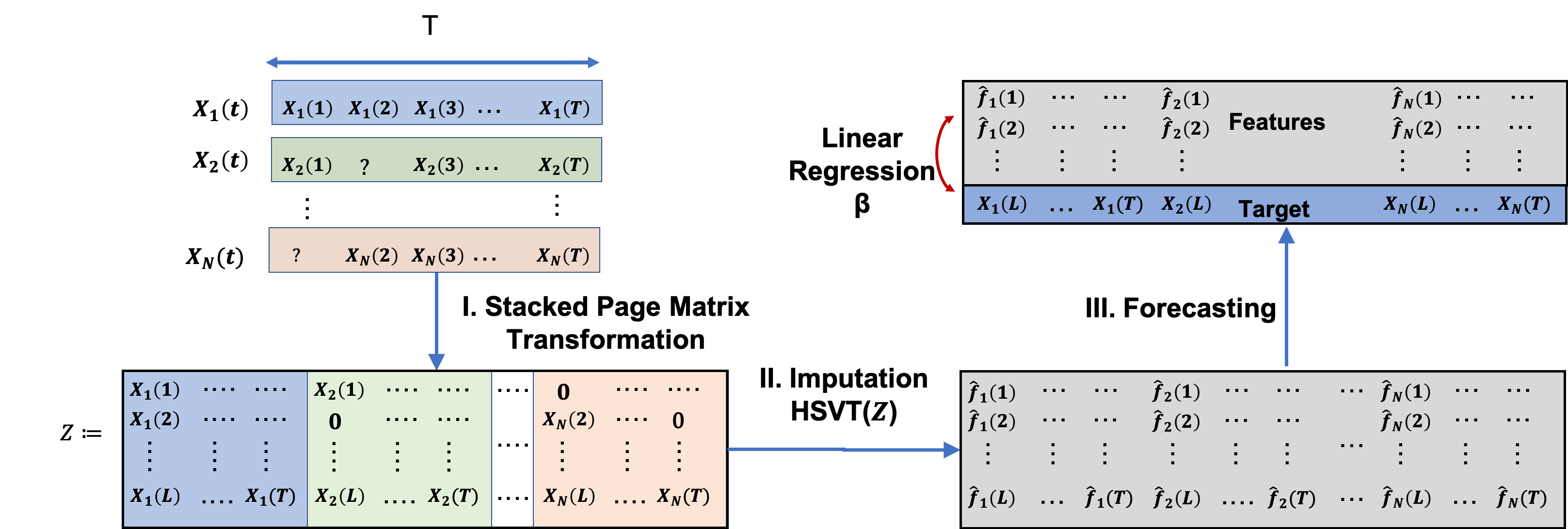}
	\caption{Pictorial depiction of the key steps of the mean estimation algorithm.}
	\label{fig:algorithm}
	\vspace{-2mm}
\end{figure}
}

\noindent \textbf{\em Imputation.} The key steps of mean imputation are as follows.\\
{
	%
	{\bf 1.}~(Form Page Matrix) 
	Transform $X_1(1:T), \dots, X_N(1:T) $ into a stacked Page matrix $\bZ^{X} \in \Rb^{L \times \bar{P}}$ with $L \leq \bar{P}$. 
	Fill all missing entries in the matrix by $0$. \\
	%
	{\bf 2.}~(Singular Value Thresholding) 
	Let SVD of  $\bZ^{X} = \bU \bS \bV^T$, where $\bU \in \Reals^{L \times L}, \bV \in \Reals^{\bar{P} \times L}$ represent left and right singular vectors and $\bS = \diag(s_1,\dots, s_L)$ the diagonal matrix of singular values $s_1\geq \dots \geq s_L \geq 0$.
	 Obtain $\bhM^{f} = \bU \bS_k {\bV}^T$ by setting all but top $k$ singular values to $0$, i.e.  $\bS_k = \diag(s_1,\dots, s_k, 0, \dots, 0)$ for some $k \in [L]$. \\
	 %
	{\bf 3.}~(Output)
	$\hat{f}_n(i+(j-1)L) \coloneqq \bhM^{f}_{i,[j+P(n-1)]} $, $i \in [L], j \in [{P}]$. \\
	%
}

%

\textbf{\em Forecasting.}
Forecasting includes an additional step of fitting a linear model on the de-noised matrix. \\
{
	%
	{\bf 1.}~(Form Sub-Matrices) 
	Let $\btZ^{X} \in \Reals^{L-1 \times \bar{P}}$ be a sub-matrix of $\bZ^{X}$ obtained by removing its last row. 
	Let $\bZ^X_L$ denote the last row. \\
	%
	{\bf 2.}~(Singular Value Thresholding)  
	Let SVD of  $\btZ^{X} = \btU \btS \btV^T$, where $\btU \in \Reals^{L-1 \times L-1}, \btV \in \Reals^{\bar{P} \times L-1}$ represent left and right singular vectors and $\btS = \diag(s_1,\dots, s_{L-1})$ the diagonal matrix of singular values $\ts_1\geq \dots \geq \ts_{L-1} \geq 0$. 
	Obtain $\bhtM^{f} = \btU \btS_k {\btV}^T$ by setting all but top $k$ singular values to $0$, i.e.  $\btS_k = \diag(\ts_1,\dots, \ts_k, 0, \dots, 0)$ for some for some $k \in [L-1]$. \\
	%
	{\bf 3.}~(Linear Regression) 
	$\hat{\beta} = \arg \min_{b \in \mathbb{R}^{L-1}} \norm{\bZ^{X}_L - (\bhtM^{f})^T b}^2_2$.\\
	%
	{\bf 4.}~(Output) 
	$\hat{f}_n(T+1) \coloneqq X_n(T-(L-1):T)^T \hat{\beta}$.
}

%

\vspace{2mm}
\textbf{Variance Estimation Algorithm.}
For variance estimation, the mean estimation algorithm is run twice, once on $\bZ^{X}$ and once on $\bZ^{X^2}$.
Then to estimate the time-varying variance (for both forecasting and imputation), a simple post-processing step is done where the square of the estimate produced from running the algorithm on $\bZ^{X}$ is subtracted from the estimate produced from $\bZ^{X^2}$. Precise details follow. 

\vspace{1mm}

\textbf{\em Imputation}. Below are steps of variance imputation. \\
{
	%
	{\bf 1.}~(Impute $\bZ^{X}, \bZ^{X^2}$) 
	Use the mean imputation algorithm on $X_1(1:T), \dots, X_N(1:T)$ (i.e., $\bZ^{X}$) and $X^2_1(1:T), \dots, X^2_N(1:T)$ (i.e., $\bZ^{X^2}$), to produce the de-noised Page matrices $\bhM^{f}$ and  $\bhM^{f^2 + \sigma^2}$, respectively. \\
	%
	{\bf 2.}~(Output) 
	Construct $\bhM^{f^2} \in \Rb^{L \times \bar{P}}$, where $\bhM^{f^2}_{ij} \coloneqq (\bhM^{f}_{ij})^2$, and produce estimates,  $\hat{\sigma}_n^2(i+(j-1)L) \coloneqq  \bhM^{f^2 + \sigma^2}_{i,[j+P\times(n-1)]}  -  \bhM^{f^2}_{i,[j+P\times(n-1)]}$, for $i \in [L], j \in [{P}]$.
	%
%
}

\vspace{1mm}

\textbf{\em Forecasting}. Just like in mean forecasting, there is an additional step after imputation.\\
{
	{\bf 1.}~(Forecast with $\btZ^{X}, \bZ^X_L, \btZ^{X^2}, \bZ^{X^2}_L $) Using the mean forecasting algorithm on $X_1(1:T), \dots, X_N(1:T) $ and $X^2_1(1:T), \dots, X^2_N(1:T) $, to produce forecast estimates $\hat{f}_n(T+1)$ and $\widehat{f_n^2 + \sigma_n^2}(T+1)$, respectively. \\
	{\bf 2.}~(Output) Produce the variance estimate 
	$\hat{\sigma}_n^2(T+1) \coloneqq \widehat{f_n^2 + \sigma_n^2}(T+1) -  (\hat{f}_n(T+1))^2.$
}
%


%

\vspace{-2mm}
\subsection{Incremental Variant of Mean and Variance Estimation Algorithms}\label{sec:inc_algorithm}
\vspace{-1mm}
\begin{figure}[!htb]
\vspace{-2mm}
	\centering
	\includegraphics[width=0.8\linewidth]{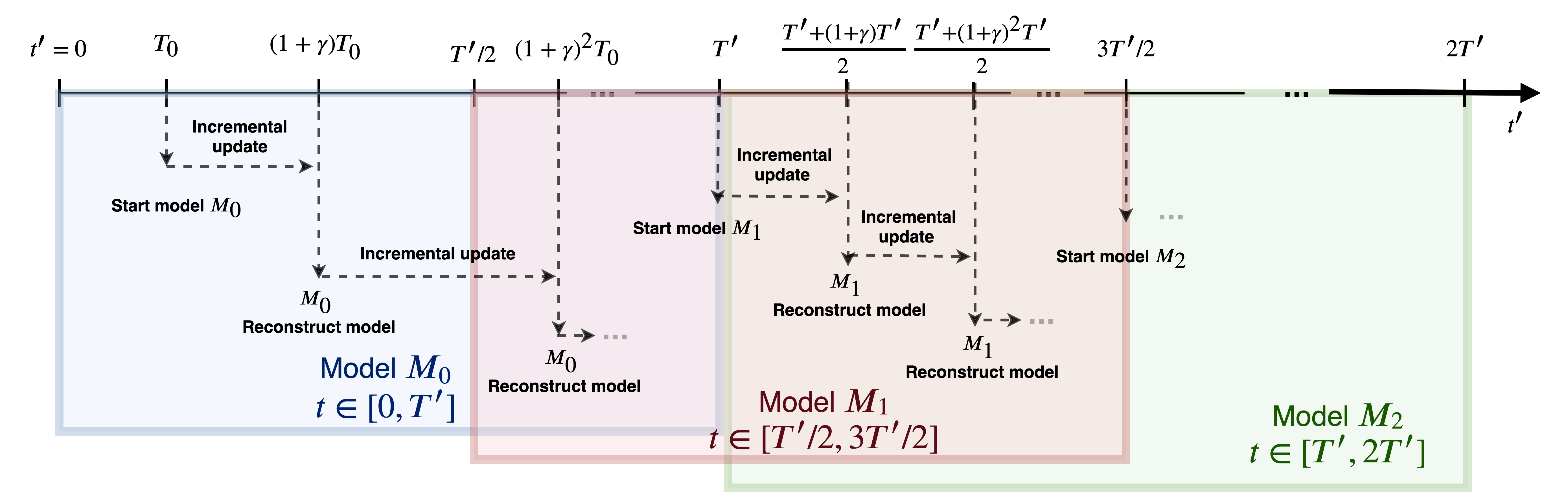}
	\caption{Incremental variant of proposed algorithm.}
	\label{fig:framework}
\vspace{-2mm}
\end{figure}
The algorithms for mean and variance estimation described above, as written, are meant for batch updating (i.e., they get to observe all data at once). 
However, for computational efficiency in any real-world application, the prediction model needs to be trained and updated incrementally, which in turn requires making these algorithms incremental. 
In particular, both the computational efficiency and statistical accuracy should not degrade with volume of data inserted.
The key computationally expensive step in the proposed algorithm is computing the SVD of the relevant Page matrices.
We address this by proposing  a simple incremental ``meta''-algorithm, which has three hyper-parameters:  $T_0, T' \in \mathbb{Z}$ and $\gamma \in (0,1] $. 
Let  $t' \coloneqq N \times t $ denote the number of observations seen thus far.
Depending on $t'$, there are three scenarios:
%


%
{
\noindent \textbf{Case 1. ($t' < T_0$)}:
%
Given that the minimum observation count has not been met, the model will output the average of the observations $X_n(0:t)$ (for both imputation and forecasting).


\noindent \textbf{Case 2. ($T_0\leq t' \leq T'$)}: 
%
\begin{inparaenum}[]
	\item {\em If} $t' =\left\lfloor T_0 (1+\gamma)^\ell\right\rfloor$ for $0\leq \ell \leq q_0$ where $q_0=\left\lfloor\dfrac{\ln(T'/T_0)}{\ln(1+\gamma)}\right\rfloor$: Re-train $M_0$ using $X_1(0:t), \dots, X_N(0:t)$ (i.e, do full batch SVD).
	\item {\em Else}: incrementally update the model $M_0$ (i.e., do incremental SVD using the method in \cite{zha1999updating}). 
\end{inparaenum}
%


\noindent \textbf{Case 3. ($ t' > T'$)}: 
%
%
\begin{inparaenum}[]
	\item Identify $M_i$ where $ i = i(t) = \max(0, \left \lfloor\frac{2t'}{T'}\right\rfloor - 1)$, and denote the first time index of $M_i$ as $s_i = \frac{iT'}{2} $. 
	\item {\em If}  $(t'-s_i) = \left\lfloor T_0 (1+\gamma)^\ell\right\rfloor$ for $0\leq \ell \leq q$ where $q=\left\lfloor\dfrac{\ln(2)}{\ln(1+\gamma)}\right\rfloor$: Re-train $M_i$ with observations $X_1(s_i:t), \dots X_N(s_i:t)$.
	\item {\em Else}: incrementally update $M_i$. 
\end{inparaenum}
}
%


Figure \ref{fig:framework} illustrates how the proposed segmentation is carried out, and at what points the model is trained fully and where it is incrementally updated.
Note, the incremental SVD method we use was developed in the Latent Semantic Indexing literature (see \cite{zha1999updating} for details).

\vspace{1mm}

{\em tspDB Algorithm Hyper-parameters}. 
The hyper-parameters of tspDB are $L, k, k_1, k_2, T_0, T', \gamma$. For all reported results in Section \ref{sec:experiments}, and in our open-source implementation, these are set in an automated manner (see  Appendix \ref{sec:hyperparameters_selection}).
In Section \ref{sec:appendix_comp_experiments}, we justify this choice of hyper-parameters.


\subsection{Algorithm Justification} \label{sec:justification}
%
%

\vspace{2mm}
\textbf{Notation.}
Let $f_n(t)$ denote the noiseless version of the time series, $X_n(t)$, i.e., $\Ex[X_n(t)] = f_n(t)$. 
Further, let $\sigma^2_n(t) = \Ex[ X_n(t)- f_n(t) ]^2$, denote the variance of the per-step noise.
Let the collection of $N$ time series, $f_1, \dots, f_N$ be denoted by $f$;
define $\sigma^2$ analogously with respect to $\sigma^2_1, \dots, \sigma^2_N$. 
Let $\bM^{f}, \bM^{\sigma^2} \in \Rb^{L \times \bar{P}}$ denote the stacked Page matrices induced by $f$ and $\sigma^2$ (analogous to how $\bZ^{X}, \bZ^{X^2}$ are defined). 
In particular, $M^{f}_{i,[j+P\times(n-1)N]} := f_n(i+(j-1)L)$ and $M^{\sigma^2}_{i,[j+P\times(n-1)N]} := \sigma^2_n(i+(j-1)L)$.

\subsubsection{Justification for Mean Estimation Algorithm}
The goal of the mean estimation algorithm is to impute and forecast the underlying time-varying mean, $f$, i.e., produce an estimate $\widehat{f}$.

\vspace{2mm}

\textbf{Linear Recurrent Formulae (LRF).}
Consider the univariate time series setting, i.e., $N = 1$.
Here, $M^{f_1}_{ij} = f_1(i+(j-1)L)$.
\begin{definition}
A time series, $f_n$, is called a order-$R$ (univariate) LRF if $i, j \in \Nb$ with $i + j \le T$,
$
	f_1(i + j) = \sum^R_{r=1} g_r(i) h_r(j),
$
where $g_r(\cdot)$, $h_r(\cdot): \Zb \to \Rb$ for $r \in [R]$.
\end{definition} 
\begin{prop} \emph{{(Proposition 5.2 in  \cite{sigmetrics})}}\label{prop:LRF_rich}
\label{prop:lowrank_LRF_example}
$
f_1(t)= \sum_{a=1}^A \exp(\alpha_a t) \cdot \cos(2\pi \omega_a t + \phi_a) \cdot P_{m_a}(t)
$
admits an order-$R$ LRF representation, where $P_{m_a}$ is any polynomial of degree $m_a$. 
Further the order $R$ of the LRF induced by $f_1$ is independent of $T$, the number of observations, and is bounded by
$
R \le A(m_{\max} + 1)(m_{\max} + 2),
$
where $m_{\max} = \max_{a \in A} m_a$.
\end{prop}
{\em Interpretation.}
From Proposition \ref{prop:LRF_rich}, we see that LRFs admit a rich class of time series families including finite sums of products of harmonics, polynomials and exponentials.
Thus LRFs well-approximate a large class of time series dynamics (e.g. stationary time series functions, see \cite{timeseries1}). 

\vspace{2mm}

\textbf{Imputation, Forecasting for LRFs.}
It is easy to see that if $f_1$ is an order-$R$ LRF, the Page matrix, $\bM^{f_1}$, induced by it is rank $R$.
In particular,
$M^{f_1}_{ij} = \sum^R_{r=1} U_{ir} V_{jr}$,
where $U_{ir} = g_r(i)$ and $V_{ir} = h_r(i)$.
Since $\bM^{f_1}$ is rank $R$ and $\Ex[\bZ^{X_1}] = \bM^{f_1}$, by recent advances in the matrix estimation and high-dimensional linear regression literature, we know that performing HSVT on $\bZ^{X_1}$ and subsequently fitting a linear model (i.e., doing Principal Component Regression) leads to a consistent estimator for both imputation and forecasting, with error scaling as $\sim \sqrt{R} / T^{1/4}$, $\sim R / \sqrt{T}$ respectively.
Specifically, see Theorem {4.1} in \cite{sigmetrics} for the imputation result; and Theorem {4.2} and Proposition {4.2} in \cite{PCR} for the forecasting result.
Note the key to these results is that the underlying matrix $\bM^{f_1}$ is (approximately) low-rank. 

\vspace{2mm}

\textbf{Multivariate LRF.}
In multivariate time series data, where $N > 1$, there exists both temporal structure, i.e., the relationship within a time series) and ``spatial'' structure, i.e., the relationship across a collection of related time series.
We propose a natural multivariate generalization of the LRF model.
Under this model, we justify below why our proposed algorithm provides accurate imputation and forecasting. 
\begin{definition}
$f$ is an order-$(K, R_{\max})$ multivariate LRF if for all $n \in [N]$,
$
f_n(t) = \sum_{k=1}^K (\theta_n)_k \cdot h_k(t),
$
where $\theta_n \in \Rb^K$ and $h_k: \Zb \to \Rb$ is an order-$R_k$ LRF,
and $R_k \le R_{\max}$ for $k \in [K]$.
\end{definition}
{\em Interpretation.}
The $(K, R_{\max})$ multivariate LRF model can be interpreted as one where there are $K$ ``fundamental'' time series $h_k(\cdot)$;
each $h_k(\cdot)$ is a (univariate) LRF, with order, $R_k$, bounded by $R_{\max}$.
Each $f_n(\cdot)$ for $n \in [N]$ can thus be seen as a unique additive combination of these $K$ LRFs,
with the latent linear parameters given by $(\theta_n)_k$.
\begin{prop}\label{prop:multi_LRF}
If $f$ is an order-$(K, R_{\max})$ multivariate LRF, then $\text{rank}(\bM^{f}) \le K \cdot R_\max$.
\end{prop}
\begin{proof}
$\bM^{f}_{i,[j+P\times(n-1)N]} 
= f_n(i+(j-1)L) = \sum_{k=1}^K (\theta_n)_k \cdot h_k(i+(j-1)L) 
= \sum_{k=1}^K (\theta_n)_k \cdot \sum^{R_k}_{r=1} U^{(k)}_{ir} V^{(k)}_{jr}$.
The last equality follows from the fact that $h_k(\cdot)$ is a order-$R_k$ LRF and so the induced Page Matrix $\bM^{h_k}$ is of rank $R_k$. Note, $U^{(k)}, V^{(k)} \in \Rb^K$ are the left and right singular vectors of $\bM^{h_k}$ respectively.
Observing that the number of terms in the summation of $ \sum_{k=1}^K (\theta_n)_k \cdot \sum^{R_k}_{r=1} U^{(k)}_{ir} V^{(k)}_{jr}$ is less than $K \cdot R_\max$, completes the proof. 
\end{proof}
{\em Interpretation.}
The key data transformation, the stacked Page matrix,  $\bM^{f}$, has rank less than $K \cdot R_\max$ and $\Ex[\bZ^{X}] = \bM^{f}$ . 
Crucially, under this model, the rank of $\bM^{f}$ is not growing with $N$ or $T$, rather just with the inherent model complexity of the multivariate time series, quantified by $K \cdot R_\max$.
{
The results of Theorem {4.1} in \cite{sigmetrics} and Theorem {4.2} and Proposition {4.2} in \cite{PCR}, suggests that the imputation and forecasting error for the algorithm we propose scales as 
{$\sim \sqrt{K\cdot R_\max}/ (NT)^{1/4}$} and {$\sim   K\cdot R_\max / \sqrt{NT}$ respectively}, 
thereby justifying the mean estimation algorithm under this model.
The additional inverse dependence on $N$ in the error scaling may be why the multivariate generalization of the time series algorithm has significantly better empirical prediction error compared to the univariate version of the algorithm. 
As stated earlier, a rigorous theoretical finite-sample analysis of our proposed algorithm is beyond the scope of this work, given our focus on building and comprehensively benchmarking tspDB.
}
\subsubsection{Justification for Variance Estimation Algorithm}
The goal of the variance estimation algorithm is to impute and forecast the underlying time-varying variance, $\sigma^2$, i.e., produce an estimate $\hat{\sigma}^2$.
A key challenge in estimating the time-varying variance of a time series is that it is not directly observed (as opposed to $f$, for which we at least get noisy, sparse observations, $X$).

\vspace{2mm}
\textbf{Multivariate LRF model for $\sigma^2$.} 
Analogous to the model we propose for $f$, we justify the variance estimation algorithm for estimating $\sigma^2$ under a multivariate LRF model.
In particular, assume $\sigma^2$ is an order-$(K^{(2)}, R^{(2)}_{\max})$ multivariate LRF; continue to assume $f$ is an order-$(K, R_{\max})$ multivariate LRF.

\vspace{2mm}
\textbf{Repeating Mean Estimation Algorithm on $\bZ^{X^2}$.}
Let $\bM^{f^2} \in \Rb^{L \times \bar{P}}$ by the stacked Page matrix induced by an entry-wise squaring of $\bM^{f}$, i.e., $M^{f^2}_{ij} := (M^{f}_{ij})^2$.
Let $\bM^{f^2 + \sigma^2} \in \Rb^{L \times \bar{P}}$ be the stacked Page matrix induced by entry-wise addition of the matrices $\bM^{f^2}$ and $\bM^{\sigma^2}$.
Then by a straightforward modification of the proof of Proposition \ref{prop:multi_LRF}, one can establish that $\text{rank}(\bM^{f^2 + \sigma^2}) \le (K \cdot R_\max)^2 + K^{(2)} \cdot R^{(2)}_{\max}$.
Note that $\Ex[X^2] = f^2 + \sigma^2$ and so $\Ex[\bZ^{X^2}] = \bM^{f^2 + \sigma^2}$.
Hence using same argument as we did in the previous section provides justification of why applying the mean estimation algorithm on $\Ex[\bZ^{X^2}]$ allows one to accurately estimate $\widehat{f^2 + \sigma^2}$ for both imputation and forecasting.
{
In particular, similar to the previous section, this suggests that the variance imputation and forecasting error will scale as {$\sim  (K \cdot R_\max) \cdot \sqrt{ K^{(2)} \cdot R^{(2)}_{\max}}/ (NT)^{1/4}$ and $\sim (K \cdot R_\max)^2 \cdot K^{(2)} \cdot R^{(2)}_{\max} / \sqrt{NT}$ respectively}.}

\vspace{2mm}

\textbf{Estimating $\sigma^2$ Under Multivariate LRF Model.}
It is then easy to see that through the following simple post-processing step, we can reliably estimate $\sigma^2$;
\[
\hat{\sigma}^2 := \widehat{f^2 + \sigma^2} - \hat{f}^2,
\] 
where $\hat{f}^2$ is the component-wise square of the estimate, $\hat{f}$, produced by the mean estimation algorithm.
If $\widehat{f^2 + \sigma^2}$ is close to $f^2 + \sigma^2$ and $\hat{f}$ is close to $f$, by triangle inequality it must be that $\hat{\sigma}^2$ is close to $\sigma^2$.
This motivates why this simple variance estimation algorithm reliably recovers the underlying (unobserved) time-varying variance under a multivariate LRF model for both $f$ and $\sigma^2$.
Empirically in Section \ref{sec:experiments_stat}, we find the simple multivariate variance estimation algorithm we propose significantly outperforms alternatives for both imputation and forecasting.
As before, a rigorous finite-sample analysis of this variance estimation algorithm is beyond the scope of this work.

\section{tspDB Implementation Details} \label{sec:tspDB_implementation_details}
\subsection{Prediction Model Storage in DB} \label{sec:system_storage}
To directly integrate tspDB with PostgreSQL, we store all model parameters in PostgreSQL itself.
For each model, the algorithms require storing: (i) the parameters associated with the appropriate truncated SVD (left/right singular vectors, and singular values); (ii) the linear regression coefficients. 
Importantly, we choose to store all of these parameters in the standard relational DB; thus, in response to a prediction query, the DB itself is queried to make predictions.
The specific schema used to store model parameters, and the relationship between them, is depicted in Figure \ref{fig:schema} (for the case $k_1 = k_2  = 3$). 
Note, the parameters associated with the variance estimation models are stored in an identical manner.
For a forecasting predictive query, we average the values of linear regression coefficients stored across models. 
Hence, we create a standard {\em materialized view} of the coefficient table to precompute the average weights of the last few models, for efficient querying.

\iftoggle{FIGURES}{
\begin{figure}[!htb]
	\centering
	\includegraphics[width=1.0\linewidth]{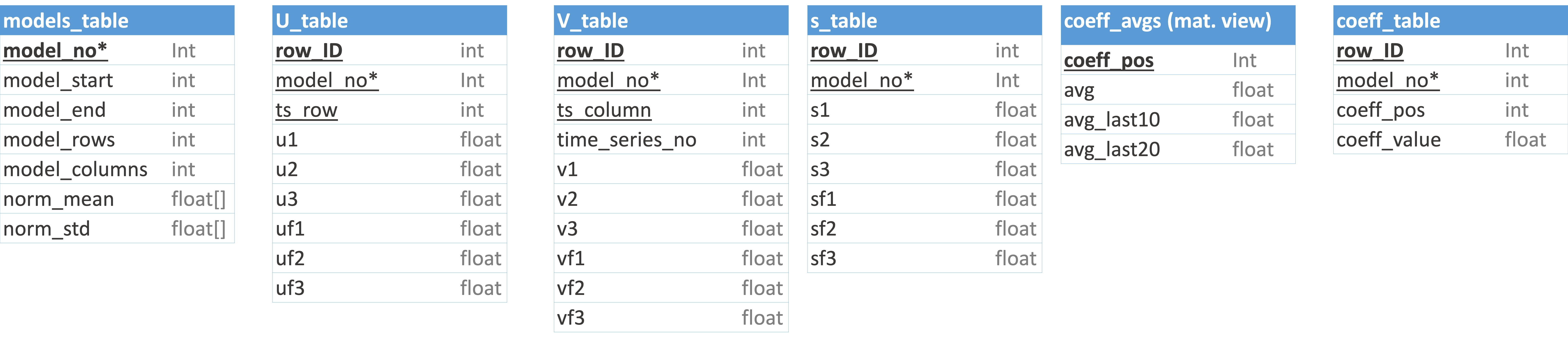}
	\caption{The schema used to store the prediction models. Note that bold attributes represent the primary keys, while underlined attributes are columns indexed by a B-tree. The columns u1, v1, s1, $\dots$ uk1, vk1, sk1) correspond to $\bZ^{X}$ and the columns  uf1, vf1, sf1, $\dots$, ufk2 ,vfk2, sfk2 correspond to $\btZ^{X}$. Refer back to Section \ref{sec:overview} for definition of $\bZ^{X}, \btZ^{X}$}
	\label{fig:schema}
\end{figure}
\vspace{-5mm} 
}

\subsection{Answering a PREDICT query in tspDB} \label{sec:system_query}
Recall the two prediction tasks supported in tspDB are imputation and forecasting. 
For both tasks, the system needs to provide a response, both of the estimated mean and the associated prediction interval (upper and lower bound) -- see Section \ref{sec:interface}. 
To answer a PREDICT query, the atomic response boils down to providing an estimation of the mean of the time series at a given $t$ with prediction interval of $c\%$ for $c \in (0, 100)$. 
This response is effectively constructed by estimating the mean, $\hat{f}_n(t)$ and the standard deviation $\hat{\sigma}_n(t)$.

\textbf{Creating a Prediction Interval.}
Using a Gaussian approximation, the $c\%$ prediction interval is given by,
\[
\left[\hat{f}_n(t) - \hat{\sigma}_n(t) \Phi^{-1}(\frac12 + \frac{c}{200}), \hat{f}_n(t) + \hat{\sigma}_n(t) \Phi^{-1}(\frac12 + \frac{c}{200})\right],
\]
where $\Phi: \Reals \to [0,1]$ denotes the Cumulative Density Function of standard Normal distribution with mean $0$ and variance $1$. 
One could alternatively utilize Chebyshev's inequality to obtain a more conservative answer as 
\[
\left[\hat{f}_n(t) - \frac{\hat{\sigma}_n(t)}{\sqrt{1-\frac{c}{100}}}, \ \hat{f}_n(t) + \frac{\hat{\sigma}_n(t)}{\sqrt{1-\frac{c}{100}}}\right].
\]
Either can be specified in tspDB.


\textbf{Answering Predictive Queries.}
Now we describe how a PREDICT query in tspDB for a given $t$ with $c\%$ prediction interval is answered. 
To start with, we determine whether it is an imputation task, i.e. $t \in [T]$ or a forecasting task, i.e. $t > T$. 
For each of these cases, we respond as follows.

\vspace{2mm}

\textbf{Imputation}: $\hat{f_n}(t): t \in [T], n \in[N]$
\footnote{Note that answering an imputation query for a given range  $\{t_1, \dots, t_2\}$ follows a similar procedure, except that it will potentially query multiple rows from \texttt{V_table}, \texttt{U_table}, and \texttt{s_table}.}. 
\renewcommand{\theenumi}{\arabic{enumi}}
\begin{inparaenum}
\itemsep0em 
    \item (Find sub-model) Let $i = i(t) = \max(0, \left\lfloor\frac{2t \times N}{T'}\right\rfloor  - 1 )$. If $i = 0$ and $t < T'/2N$, use $\cI(t) = \{0\}$ else $\cI(t) = \{i(t), i(t)+1\}$.  
    \item (Find Row, Column Indices) Let $\trow(j) = \big(t - \frac{jT'}{2N}\big) \mod L$, $\tcol(j) = N \Big\lfloor \big(t - \frac{jT'}{2}\big) / L \Big\rfloor + (n-1)$ be row, column indices of matrix corresponding to the 
    models with $j \in \cI(t)$. 
    
	\item (Find Truncated SVD for Mean) Query left, right singular vectors $U^j, V^j$ respectively from \texttt{U_table}, \texttt{V_table} for model corresponding to mean values with index 
	$j \in \cI(t)$ along with singular values $S^j$ from \texttt{s_table}. 
	
	\item (Produce Mean Estimate) Set: 
	$\hat{f_n}(t) = \frac{1}{|\cI(t)|} \sum_{j \in \cI(t)} \sum_k U^j_{\trow(j) k} V^j_{\tcol(j) k} S^j_{k}.$

	\item (Find Truncated SVD for Second Moment) Query left, right singular vectors $\tU^j, \tV^j$ respectively from \texttt{V_table}, \texttt{U_table} for model corresponding to second moment with index $j \in \cI(t)$ along with singular values $\tS^j$ from  \texttt{s_table}. 
	
	\item (Produce Variance Estimate) Set $\hat{\sigma_n}^2(t) = \max(0, \widehat{f_n^2 + \sigma_n^2}(t) - \hat{f_n}(t)^2)$, where:
	
	$\widehat{f_n^2 + \sigma_n^2}(t)  = \frac{\sum_{j \in \cI(t)} \sum_k \tU^j_{\trow(j) k} \tV^j_{\tcol(j) k} \tS^j_{k}.}{|\cI(t)|} $.
	
	\item (Output Prediction Interval) Output interval using $\hat{f_n}(t), \hat{\sigma_n}(t)$ for queried confidence $c\%$.

\end{inparaenum}

\vspace{2mm}

\textbf{Forecasting}:  $\hat{f_n}(t): t > T,  n \in[N]$. 
\renewcommand{\theenumi}{\arabic{enumi}}
\begin{inparaenum}
\itemsep0em 
\item (Retrieve History) Query the last L-1 observations $X_n(T - L+1:T))$.
\item For $t_0 \in \{T -L+1, \dots, T\}$, set $g^m_n(t_0) = X_n(t_0)$ if $X_n(t_0)$ is not missing, otherwise set it to zero. Similarly, set  $g^v_n(t_0) = X^2_n(t_0)$
\item (Obtain Coefficients) Obtain a certain  coefficients average from the materialized view (e.g. average of last 10 models) 
for means $\hat{\beta}^m$ and variances $\hat{\beta}^v$. 
\item (Sequential Forecasting) For $\tau \in \{T+1, \dots, t\}$ produce estimate of means and variances 
as
$ g^m_n(\tau) = \sum_{\ell=1}^{L-1} g^m_n(\tau-\ell) \hat{\beta}^m_\ell, ~\text{and}~
g^v_n(\tau) = \sum_{\ell=1}^{L-1} g^v_n(\tau)(\tau-\ell) \hat{\beta}^v_\ell$.

\item (Output Prediction Interval) Output estimate $\hat{f_n}(t) = g^m_n(t) $ and its  prediction interval using $\hat{\sigma_n}(t) = g^v_n(t) $ 
for queried confidence $c\%$.
\end{inparaenum}


\vspace{-4mm}
\section{Statistical and Computational Benchmarking of tspDB}\label{sec:experiments}
In this section, we detail the extensive testing we conduct to benchmark tspDB's statistical and computational performance against popular state-of-the-art prediction libraries, and just its computational performance against PostgreSQL. 
We recall the main statistical and computational results from this section are summarized in Table \ref{table:full} in Section \ref{sec:introduction}.
In Section \ref{sec:set_up}, we detail the experimental setup, e.g., datasets, machine configuration, algorithm hyper-parameters used. 
In Section \ref{sec:experiments_stat} and \ref{sec:computational_experiments}, we extensively benchmark tspDB's statistical and computational performance against other state-of-the-art time series prediction methods. 
In Section \ref{sec:postgres_computational_experiments}, we benchmark tspDB's computational performance against PostgreSQL with respect to standard DB metrics.

\vspace{-3mm}
\subsection{Setup}\label{sec:set_up}

\textbf{Datasets.} 
Throughout the experiments, we use three real-world datasets that are standard benchmarks (e.g. used
in \cite{TRMF}) in time series analysis as well as three synthetic datasets. 
We use these synthetic datasets so we can compare with the latent mean and variance which are of course not observable in real-world data.
The real-world datasets come from three domains: electricity \cite{ucielec}, traffic \cite{ucitraffic} and finance \cite{WRDS}. 
For further details on the datasets used, refer to Appendix \ref{sec:data_desc}.

\vspace{1mm}

\textbf{Machine and DB Configuration.} \label{sec:setup}
In all experiments, we use an Intel Xeon  E5-2683 machine with 16 cores, 132 GB of RAM, and an SSD storage
\footnote{Note all experiments were replicated on different machine configurations as well by replacing PostgreSQL by TimeScaleDB \cite{timescale}. We find the metrics of interest remain similar. We omit these results due to space constraints.}.
In Table \ref{DBSpec} in Section \ref{sec:system_config},  we detail the relevant settings used for PostgreSQL 12.1. 

\vspace{1mm}

\textbf{Algorithms Setup and Hyper-parameters.} 
Throughout the experiments, we compare with several state-of-the-art algorithms. Including  LSTM \cite{LSTM},  DeepAR \cite{DeepAR},  TRMF \cite{TRMF}, and Prophet \cite{Prophet}.
In all reported results, tspDB's hyper-parameters  ($L, k, k_1, k_2, T_0, T', \gamma$) are set in an automated manner, as detailed Appendix \ref{sec:hyperparameters_selection}.
Refer to Appendix \ref{sec:exp_appendix_methods} for more information about the implementations and parameters used for all other methods.

\vspace{1mm}

\textbf{Accuracy Metrics.} 
As stated earlier, for each experiment, we use Normalized Root-Mean-Square-Error (NRMSE) as the accuracy metric, where normalization corresponds to rescaling each time series to have zero mean and unit variance before calculating the RMSE. 
Further, to quantify the statistical accuracy of the different methods used, we use a variant of the standard Borda Count (weighted by the NRMSE in each experiment), which we denote as WBC. 
As explained in Appendix \ref{sec:score_metric}, WBC lies between $0$ and $1$: 
for a given algorithm, $0.5$ means it does as well as all other algorithms on average over all experiments; $1$ (resp. $0$) means it does `infinitely' better (resp. worse). 
Again, we emphasize that by simply looking at the mean NRMSE as is normally done, tspDB's relative performance is even better (see Table \ref{table:stat}).
We use WBC as we believe it better summarizes the statistical accuracy of the various algorithms used across the various experiments.

{
\vspace{-2mm}
\subsection{Statistical Benchmarking of tspDB}\label{sec:experiments_stat}
}
\begin{table}[]
\fontsize{8.9pt}{8.9pt}\selectfont
\tabcolsep=0.11cm
\centering
\caption{tspDB outperforms state-of-the-art algorithms in mean and variance estimation across different datasets. We use WBC (higher is better) and average NRMSE (lower is better)  to summarize the results.  }
\label{table:stat}
\begin{tabular}{@{}ccccccccc@{}}
\toprule
        & \multicolumn{4}{c}{\begin{tabular}[c]{@{}c@{}} Mean Imputation \\ (WBC/NRMSE)\end{tabular}}                                        & \multicolumn{4}{c}{\begin{tabular}[c]{@{}c@{}}Mean Forecasting  \\ (WBC/NRMSE)\end{tabular}}                                  \\  \cmidrule[0.5pt](lr){2-5} \cmidrule[0.5pt](lr){6-9} 
        & \multicolumn{1}{c}{Electricity} & \multicolumn{1}{c}{Traffic} & \multicolumn{1}{c}{Synthetic I} & \multicolumn{1}{c}{Financial} & \multicolumn{1}{c}{Electricity} & \multicolumn{1}{c}{Traffic} & \multicolumn{1}{c}{Synthetic I} & \multicolumn{1}{c}{Financial}                                                     \\ \midrule
tspDB    & \textbf{0.61}/\textbf{0.39}                             & \textbf{0.50}/\textbf{0.49}                        & \textbf{0.53}/\textbf{0.25}                             & \textbf{0.65}/\textbf{0.28}              &              \textbf{0.53}/ \textbf{0.48} &  {0.50}/0.53  &  \textbf{0.70}/\textbf{0.20}  &  \textbf{0.66}/\textbf{0.36}                                                        \\
LSTM    & NA                              & NA                          & NA                              & NA                            &0.49/0.55 &  \textbf{0.54}/\textbf{0.47} &  0.45/0.44 &  0.31/1.20                                                                                           \\
DeepAR  & NA                              & NA                          & NA                              & NA                          & \textbf{0.53}/\textbf{0.48} &  0.53/\textbf{0.47} &  0.53/0.33 &  0.65/0.40                                                                                           \\
TRMF    & 0.39/0.69                            &  \textbf{0.50}/0.51                        & 0.47/0.33                            & 0.35/0.51                      & 0.50/0.53 &  0.48/0.57 &  0.61/0.27 &  0.58/0.46        \\
Prophet & NA                              & NA                          & NA                              & NA                            &   0.47/0.58 &  0.45/0.62 &  0.22/1.01 &  0.30/1.30                                                                                                   \\                                                                                     
\bottomrule
\end{tabular}
\end{table}
%


\textbf{Mean Imputation.} 
{
We compare the imputation accuracy of tspDB against TRMF, a popular method for imputing (and forecasting) multivariate time series data. 
We choose it as the single point of comparison as in the work of \cite{impsurvey}, the authors conduct extensive benchmarking of various time series imputation methods, and find TRMF to have state-of-the-art statistical and computational performance. 
We note the other time series libraries we compare against, e.g., LSTM, DeepAR, do not have imputation functionality.
}  

\hspace{4mm} 
We evaluate both tspDB and TRMF using the three real-world datasets and one synthetic dataset (synthetic I) introduced in Section \ref{sec:set_up}. 
To test the algorithms' robustness, we corrupt the various datasets by artificially masking entries with probability $(1-p)$, and by adding zero-mean Gaussian noise to each entry with varying standard deviation $\sigma$.
As we vary $p$ and $\sigma$ for all four datasets, we find tspDB outperforms TRMF in $80$\% of experiments.
In  Table \ref{table:stat}, we summarize the statistical performance of tspDB against TRMF  in terms of  WBC  and the average NRMSE for each dataset across all experiments.
We find that using both metrics,  tspDB outperforms TRMF on all datasets. 
%

\vspace{1mm}%
\textbf{Mean Forecasting.} 
We compare the forecasting performance of tspDB  against LSTM,  DeepAR, TRMF, and Prophet---these are some of the most popular time series libraries used in academia and industry.
Using the same datasets used for mean imputation and varying $p$ and $\sigma$ in the same way, we find that tspDB performs competitively with deep-learning algorithms (DeepAR and LSTM), and outperforms both TRMF and Prophet.  
Specifically, as we vary the fraction of missing values and the noise added, tspDB is the best performing method in 50\% of experiments, and at least the second best performing in $80\%$ of experiments. 
To summarize the statistical performance of all algorithms, we compute their WBC and the average NRMSE for each dataset in Table \ref{table:stat}; we find tspDB, using both metrics,  outperforms all other algorithms in the electricity,  financial and synthetic I datasets, while performing competitively with DeepAR and LSTM in the traffic dataset.

\vspace{1mm}
\textbf{Variance Estimation.} 
%
We restrict our analysis to synthetic data as we do not get access to the true underlying time-varying variance in real-world data.
We use the dataset synthetic II, which consists of nine sets of multivariate time series each with a different additive combination of time series dynamics and a different noise observation model (Gaussian, Poisson, Bernoulli noise). See Appendix \ref{sec:data_desc} for details on how synthetic II was generated. 
We again vary the fraction of observed data {$p \in \{1.0, 0.8, 0.5\}$}. 
For imputation, to the best of our knowledge, there is no algorithm which produces prediction intervals.
Hence we use an adaptation of TRMF to benchmark tspDB's performance---note though it does not natively support variance estimation. 
For forecasting, we compare with DeepAR as it has in-built functionality to produce confidence intervals for its predictions.
Refer to Appendix \ref{sec:exp_appendix_methods} for details of how we use these algorithms to do variance estimation.

\hspace{4mm}  We find tspDB has superior performance in all but one experiment ($>98\%$) over both the adapted version of TRMF (for imputation) and DeepAR's in-built functionality (for forecasting).
See Table \ref{table:var_imp} for the NRMSE values for each experiment along with a summary of the performance of each algorithm using WBC and the mean NRMSE across experiments.  
In summary, we find tspDB outperforms TRMF in variance imputation and  DeepAR  in variance forecasting across both metrics (WBC, mean NRMSE).
Specifically, with respect to imputation, we find tspDB outperforms TRMF in all but one experiment,  where the ratio of TRMF's error to tspDB's in the range of $0.97$-$2.27$ (see Table \ref{table:var_imp}). 
With respect to forecasting, we find tspDB outperforms DeepAR in all experiments, where the ratio of DeepAR's error  to tspDB's is in the range $1.01$-$1870.59$ (see Table \ref{table:var_imp}).
tspDB's performance is notably superior when dealing with integer (e.g., Poisson generated) observations.  
Collectively, these experiments show the robustness of tspDB, in that it is  “noise model agnostic” when estimating the mean and variance of a time series.

\begin{table}[h]
\centering
\caption{ tspDB outperforms both TRMF (in variance imputation) and DeepAR (in variance forecasting).}
\label{table:var_imp}
\fontsize{7.5pt}{8.0pt}\selectfont
\tabcolsep=0.03cm
\begin{tabular}{@{}lllllllllllllll@{}}
\toprule
\multicolumn{1}{l}{\multirow{2}{*}{\begin{tabular}[c]{@{}c@{}}Observation \\ Model\\ \end{tabular}}} & \multicolumn{1}{l}{\multirow{2}{*}{\begin{tabular}[c]{@{}c@{}}Time Series \\ Dynamics\\ \end{tabular}}}       & \multicolumn{3}{c}{ \begin{tabular}[c]{@{}c@{}}tspDB \\ (Imputation)\\ \end{tabular}}  & \multicolumn{3}{c}{ \begin{tabular}[c]{@{}c@{}}TRMF \\ (Imputation)\\ \end{tabular}} & \multicolumn{3}{c}{ \begin{tabular}[c]{@{}c@{}}tspDB \\ (Forecasting)\\ \end{tabular}}  & \multicolumn{3}{c}{ \begin{tabular}[c]{@{}c@{}}DeepAR \\ (Forecasting)\\ \end{tabular}}  \\\cmidrule[0.5pt](lr){3-5} \cmidrule[0.5pt](lr){6-8}  \cmidrule[0.5pt](lr){9-11} \cmidrule[0.5pt](lr){12-14}  
\multicolumn{2}{l}{}                                       & p = 1.0      & p = 0.8     & p = 0.5  & p = 1.0      & p = 0.8     & p = 0.5    & p = 1.0      & p = 0.8     & p = 0.5    & p = 1.0      & p = 0.8     & p = 0.5        \\ \midrule
\multirow{3}{*}{Gaussian}               & Har                     & \textbf{0.076}          & \textbf{0.099} &\textbf{0.118}    & 0.122 &0.125&0.141 & \textbf{0.106}&\textbf{0.144}& \textbf{0.156}&  0.170&0.184&0.289\ \\
             					   & Har +  trend                &\textbf{0.075} & \textbf{0.091} & \textbf{0.103}   &0.133& 0.135& 0.142  &\textbf{0.154}&\textbf{0.155}&\textbf{0.247 }  &0.286&0.232&0.269  \\
  						  & Har+AR+trend & \textbf{0.074}          & \textbf{0.090}& \textbf{0.101}      & 0.134 &0.136&0.146 & \textbf{0.173}&\textbf{0.263}&\textbf{0.337} &0.214&0.265&0.388 \\  \midrule

\multirow{3}{*}{Poisson}   
& Har &\textbf{0.126}&\textbf{0.132} &\textbf{0.150}&0.137&0.138&0.151   & \textbf{0.143}&\textbf{0.152}&\textbf{0.157 }         &13.20 &148.5 &     90.20             \\
 & Har+ trend     & \textbf{0.086}&\textbf{0.087}&\textbf{0.101}    &0.176& 0.187&0.194    & \textbf{0.093}&\textbf{0.182}&\textbf{0.199} &    0.491 &1.403&2.163      \\
 & Har+AR+trend &\textbf{0.081}&\textbf{0.088}&\textbf{0.104}   &0.184&0.187&0.204 & \textbf{0.093}&\textbf{0.182}&\textbf{0.199} &     1.386& 34.45&162.5   \\  \midrule
 
\multirow{3}{*}{Bernoulli} 
& Har    & \textbf{0.024}&\textbf{0.027 }  & 0.030   & 0.026 &0.027&\textbf{0.029}  & \textbf{0.029}   &\textbf{0.050}&\textbf{0.033 }   &0.073 & 0.072&      0.077       \\
& Har +  trend    &\textbf{0.022}&\textbf{0.025}&\textbf{0.025 }&0.030&0.030&0.032  & \textbf{0.029}&\textbf{0.048}&\textbf{0.059}   &0.049&0.068&0.076   \\
& Har+AR+trend & \textbf{0.022}&\textbf{0.024}&\textbf{0.025 } &0.030 &0.030&0.032& \textbf{0.036}&\textbf{0.036}&\textbf{0.056}         &    0.070 &0.082&0.111   \\  \cmidrule[1.2pt](){1-14}  
\multirow{2}{*}{Summary} & WBC  &\multicolumn{3}{c}{\textbf{0.597}}  &\multicolumn{3}{c}{0.403}   &\multicolumn{3}{c}{\textbf{0.726}}  &\multicolumn{3}{c}{0.264}    \\ 
& Mean NRMSE &\multicolumn{3}{c}{\textbf{0.070}}  &\multicolumn{3}{c}{{0.112}}   &\multicolumn{3}{c}{\textbf{0.132}}  &\multicolumn{3}{c}{{16.9}}      
\\\cmidrule[1.2pt](){1-14}  
\end{tabular}
\vspace{-4mm}
\end{table}

\subsection{tspDB's Robustness to Noise Models} \label{sec:exp_noise}
In practice, the same latent time series dynamic can lead to very different observations depending on the type of noise that is present (e.g., Gaussian, Poisson, Bernoulli). 
%
%
Thus, as an additional experiment, we showcase tspDB's robustness against different noise models.
Pleasingly, we find tspDB is ``noise agnostic'', i.e., it effectively imputes and forecasts the latent time-varying mean {\em without knowledge of the underlying noise model}.

\hspace{4mm}  For these set of experiments, we use dataset synthetic III (see Appendix \ref{sec:data_desc} for details);
we generate three noise models---Gaussian (i.e., float) , Bernoulli (i.e., boolean), and Poisson (i.e., integer)---all with the same latent  time series dynamics (normalized to lie within the interval $[0, 1]$) as shown in Figures \ref{fig:noise1}, \ref{fig:noise2}, \ref{fig:noise3},  respectively.
We find that tspDB's imputation error in RMSE/$R^2$ is ($0.079$/$0.854$), ($0.078$/$0.858$) and ($0.102$/$0.758$) for the Gaussian, Bernoulli and Poisson observation models respectively.
Similarly, the forecasting error in RMSE/$R^2$ is  ($0.041$/$0.979$), ($0.083$/$0.914$) and ($0.110$/$0.851$) for the Gaussian, Bernoulli and Poisson observation models respectively.
Figure \ref{fig:noise} gives a visual depiction of how tspDB effectively imputes and forecasts under Gaussian, Bernoulli and Poisson observation models. 
%

\iftoggle{FIGURES}{
\begin{figure}[!htb]
\vspace{-2mm} 
\floatconts
{fig:noise}
{\caption{ Without knowledge of whether the observations come from a time-varying Gaussian process (i.e. floats, see Figure \ref{fig:noise1}), a Bernoulli process (i.e. binary data,  see Figure \ref{fig:noise2}), or a Poisson process (i.e. integers,  see Figure \ref{fig:noise3}),  tspDB successfully imputes and forecasts the underlying mean.}}
{
	\subfigure[][t]{\label{fig:noise1}
		\includegraphics[width=0.30\linewidth]{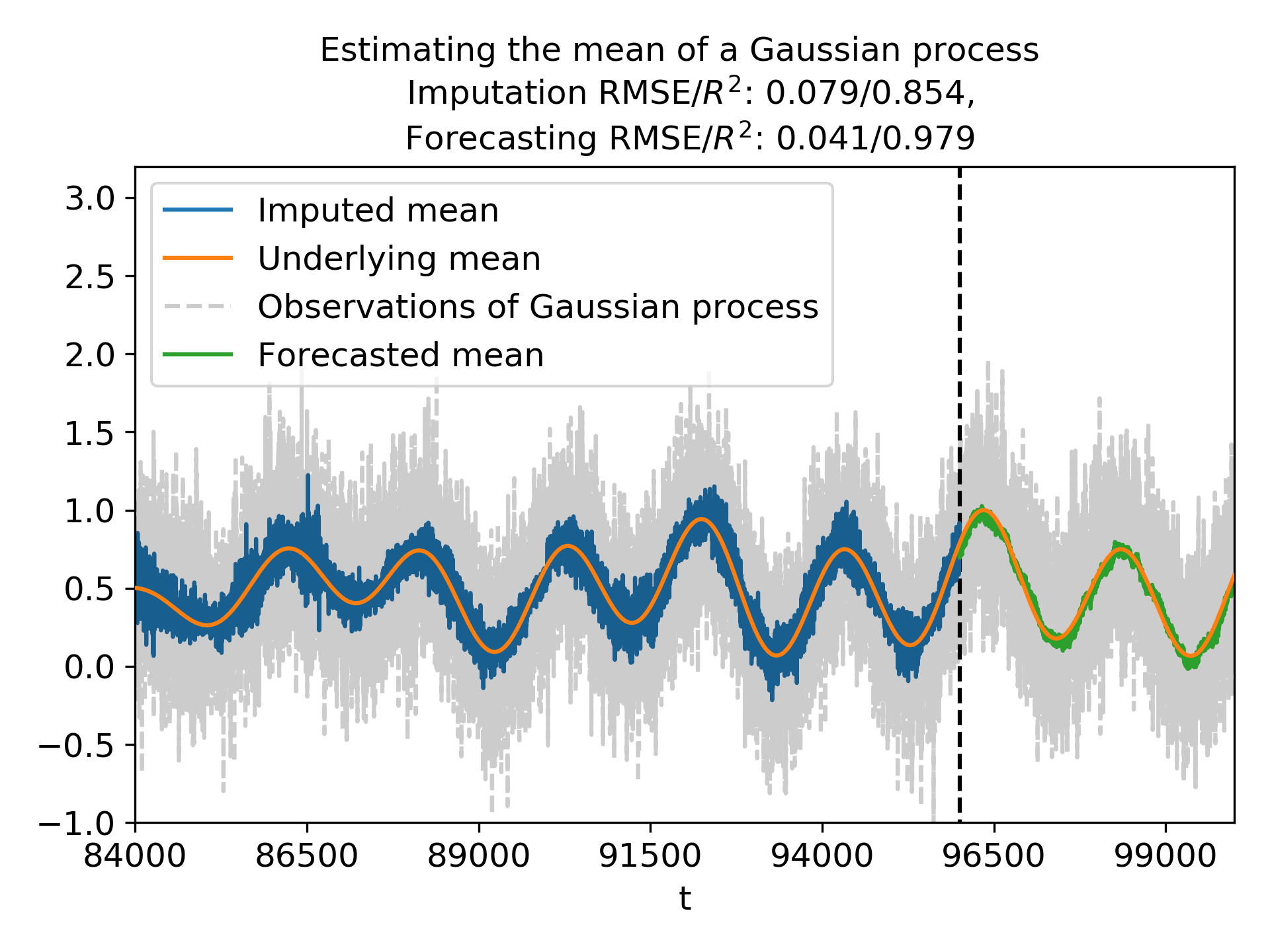}}
		\quad
	\subfigure[][t]{\label{fig:noise2}
		\includegraphics[width=0.30\linewidth]{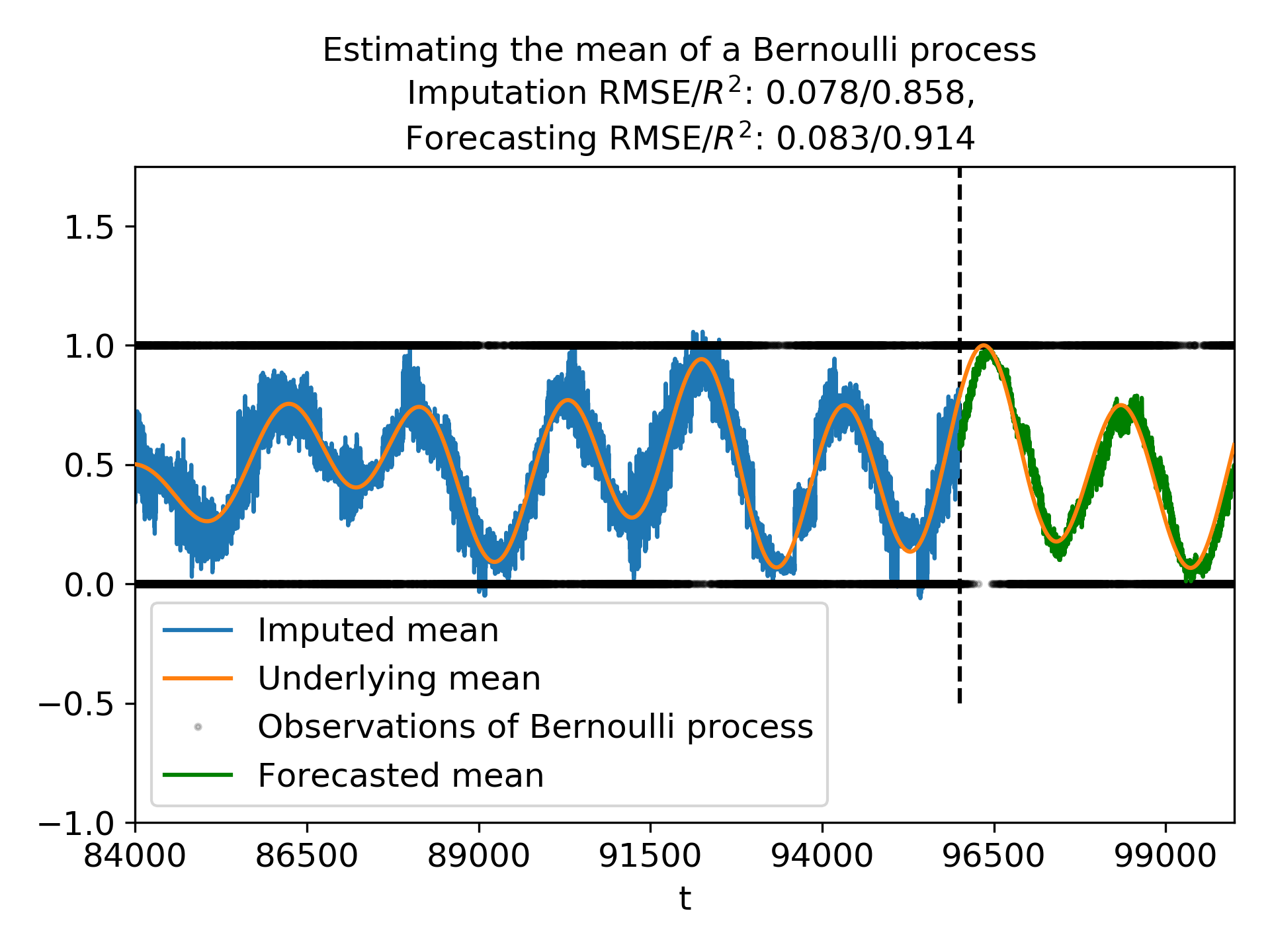}}
		\quad
	\subfigure[][t]{\label{fig:noise3}
		\includegraphics[width=0.30\linewidth]{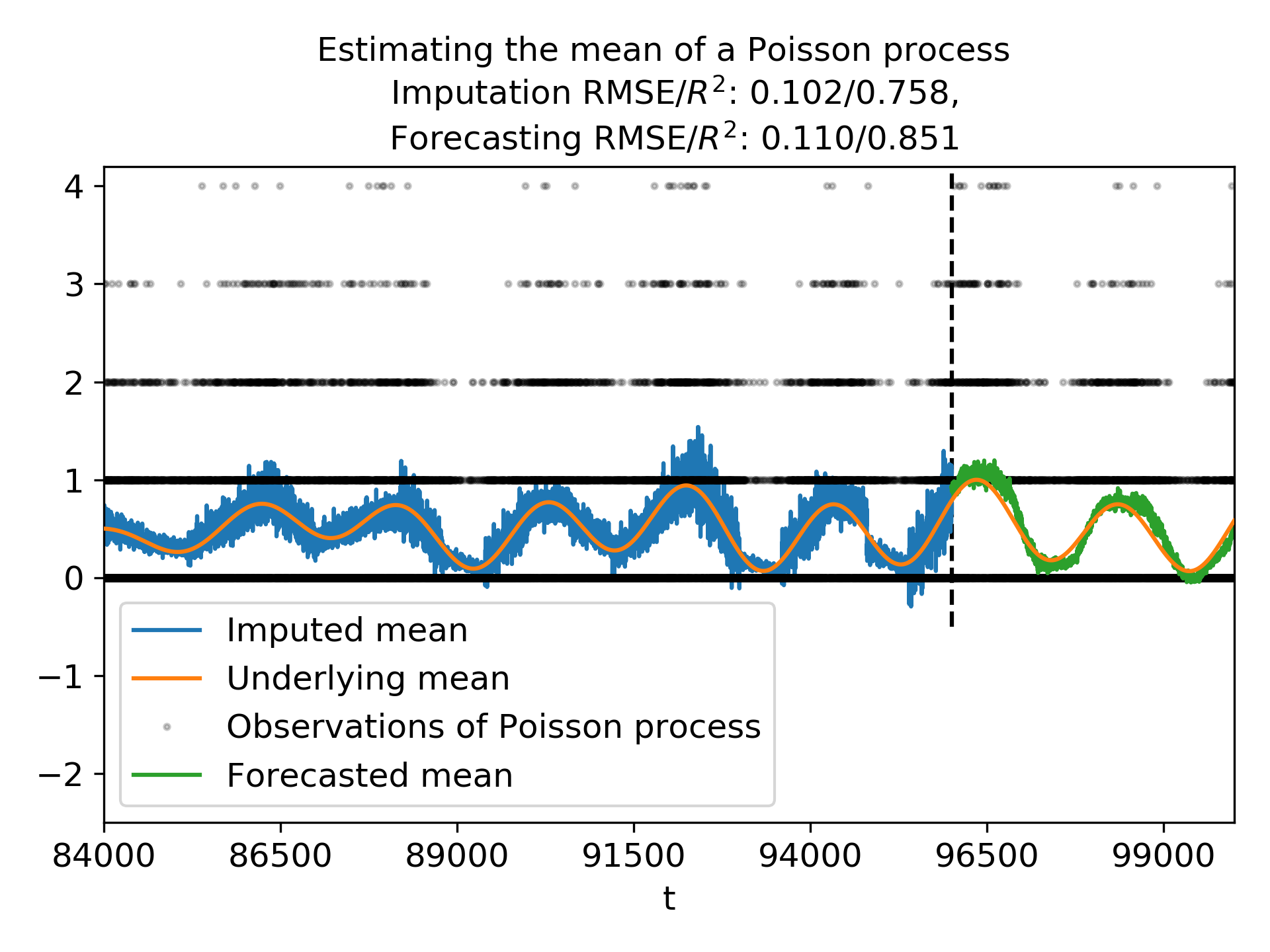}}
	\label{fig:noise}
	\vspace{-5mm} 
}
\vspace{-5mm} 
\end{figure}
}

\vspace{-3mm}
{
\subsection{Computational Benchmarking of tspDB vs. Other Prediction Algorithms}\label{sec:computational_experiments}
}
\textbf{tspDB vs. Other Prediction Methods.}  

\begin{table}[b]
\fontsize{8.9pt}{8.9pt}\selectfont
\caption{tspDB has significantly less training time, prediction query latency compared to  alternatives.}
\label{table:comp}
\centering
\begin{tabular}{@{}lllllllll@{}}
\toprule
           & \multicolumn{4}{c}{\begin{tabular}[c]{@{}c@{}}Training/Insert Time \\ (seconds)\end{tabular}} & \multicolumn{4}{c}{\begin{tabular}[c]{@{}c@{}}Prediction/Querying Time \\ (milliseconds)\end{tabular}} \\ \cmidrule[0.5pt](lr){2-5} \cmidrule[0.5pt](lr){6-9} 
           & Electricity             & Traffic             & Synthetic I            & Financial            & Electricity               & Traffic               & Synthetic I               & Financial              \\ \midrule
tspDB       & \textbf{11.8}                   & \textbf{14.4}               & \textbf{6.2}                   & \textbf{2.3}                  & \textbf{3.3}                       & \textbf{3.9}                   &  3.2                   &      \textbf{3.6}                 \\
LSTM       & 1357.5                  & 661.2               &      450.9             & 97.9                 & 295.6                     & 358.3                 & 278.7                      &  354.8                \\
DeepAR     & 907.6                   & 572.2               &   491.2                 &      144.5          & 366.6                     & 378.5                 &   285.8                 &        207.2             \\
TRMF       & 34.5                    & 14.8                &  12.4                   &  3.6                &       7.3          &      6.0            &            \textbf{2.5}            &   4.6                 \\
Prophet    & 6868.8                 & 4726.3             &    535.9              &         2165.4        & 1473.4                    & 1458.8                & 1480.3                     &     1434.0            \\
\bottomrule 
\end{tabular}
\vspace{-4mm}
\end{table}

{
As stated earlier, we evaluate the computational performance of all algorithms using two metrics:
(i) the time it takes to train a prediction model on a given dataset---each of the methods we benchmark against exposes a ``.fit()'' interface or equivalent (similar to CREATE prediction_model in tspDB), and we measure the time it takes for a prediction model to be returned from when ``.fit()'' is executed;
(ii) the time required by the model to produce a point forecast for a time series---each of the methods we benchmark against exposes a ``.predict()'' interface or equivalent (similar to a PREDICT query in tspDB), and we measure the time it takes for a fitted prediction model to return a forecast for a queried future time step (e.g., the time required to do a forward pass in an LSTM network to predict the stock price tomorrow for a certain company).
We evaluate all algorithms using the machine configuration detailed in Section \ref{sec:set_up}.  
The hyper-parameters chosen for the various algorithms are detailed in Appendix \ref{sec:exp_appendix_methods}; just like for tspDB, we choose the default parameters in the open-source implementations.
}

\hspace{4mm} In Table \ref{table:comp}, we report results for these experiments. 	
With respect to time taken to train a model, we find tspDB is {$42.4$-$114.9$x} faster than LSTM, {$39.7$-$79.2$x} faster than DeepAR,  {$1.03$-$2.92$x} faster than TRMF,  and {$86.4$-$832.7$x} faster than Prophet. 
%
With respect to prediction queries, tspDB's latency is  {$86.0$-$99.4$x} faster than LSTM's,  {$58.0$-$110.4$x} faster than DeepAR's, {$0.8$-$2.2$x} the latency of TRMF's,  and {$370.3$-$456.9$x} faster than Prophet's.

\vspace{1mm}

\subsection{Computational Benchmarking vs. PostgreSQL}\label{sec:postgres_computational_experiments}

\vspace{-1mm}
\textbf{tspDB vs. PostgreSQL.} 
On the same datasets we use to benchmark tspDB against these other algorithms, we compare: 
(i) tspDB's training time (i.e. CREATE predict model query) vs. PostgreSQL's bulk insert time;
(ii) tspDB's PREDICT query latency against PostgreSQL's SELECT query latency.
We find tspDB's model training time ranges from {$0.58$-$1.52$x} to that of the insert time of PostgreSQL. 
We evaluate the prediction query latency with and without uncertainty quantification (UQ), i.e., producing prediction intervals via variance estimation---see Section~\ref{sec:system_query} for details on how we do UQ. 
Compared with a SELECT query in PostgreSQL, imputation queries are  {$1.64$-$2.67$x} slower, and 
forecasting queries are  {$1.67$-$2.77$x} slower.
If UQ is added, queries are {$3.34$-$5.35$x} and {$3.42$-$5.48$x} slower for imputation and forecasting, respectively.
This amounts to $1.29$-$2.36$ milliseconds for SELECT queries, $3.22$-$3.87$ milliseconds for imputation 
queries ($5.65$-$7.88$ milliseconds with UQ), and  $3.24$-$3.94$ milliseconds for forecasting queries ($5.67$-$8.08$ milliseconds with UQ).
Refer to Table \ref{table:comp_postgres} for a summary of results. 


\begin{table}[!th]
\vspace{-2mm}
\centering
\fontsize{8.9pt}{8.9pt}\selectfont
\caption{$N$ and $T$ refer to number of time series and number of observations per time series respectively.}
\label{table:comp_postgres}
\begin{tabular}{@{}lllllll@{}}
\toprule
            & \multicolumn{1}{c}{\multirow{2}{*}{}} & \multicolumn{2}{c}{\begin{tabular}[c]{@{}c@{}}Training/Insert Time\\ (seconds)\end{tabular}} & \multicolumn{3}{c}{\begin{tabular}[c]{@{}c@{}}Query Latency\\ (milliseconds)\end{tabular}}                                                                                                                          \\ \cmidrule[0.5pt](lr){3-4} \cmidrule[0.5pt](lr){5-7} 
            & \multicolumn{1}{c}{$N / T$}                       & tspDB                                       & PostgreSQL                                      & \begin{tabular}[c]{@{}l@{}}tspDB's Forecast  \\ /with UQ \end{tabular} & \begin{tabular}[c]{@{}l@{}}tspDB's Imputation\\ /with UQ\end{tabular} & \begin{tabular}[c]{@{}l@{}}PostgreSQL's\\  SELECT\end{tabular} \\ \midrule
Electricity & 370 / 25968                                & 11.81                                      & 8.34                                            & 3.32/7.02                                                                & 3.25/6.93                                                              & 1.34                                                           \\
Traffic     & 963 / 10392                                & 14.42                                      & 9.50                                            & 3.94/8.08                                                             & 3.87/7.88                                                             & 2.36                                                           \\
Synthetic I & 400 / 15000                                &      6.20                                   & 4.93                                            & 3.24/5.67                                                                 & 3.22/5.65                                                             & 1.31                                                           \\
Financial   & 839 / 3993                                 & 2.31                                        & 4.00                                            & 3.57/7.07                                                                & 3.47/6.90                                                              & 1.29                                                           \\ \bottomrule
\end{tabular}
\vspace{-2mm}
\end{table}

\textbf{Scalability of tspDB.} 
We compare how the computational performance of tspDB relative to PostgreSQL scales with respect to the metrics we consider as we vary the amount of data inserted. 
{We generate a single time series that is a sum of harmonics and then corrupt it with additive Gaussian noise.
Specifically, we generate $X(t) = f(t) + \epsilon(t)$, where $\epsilon(t) \sim \mathcal{N}(0, 0.1)$  and  $f(t) = \sum_{i=1}^4 \alpha_i \cos(\omega_i t/T) $, where $t \in [T]$,  $\alpha_i \in [-1.5,1.5]$ and $\omega_i \in [1,100]$ (randomly selected).
We vary the amount of data points $(T)$ between {$10^4$ to $10^8$}.
}
Using this time series data, we create a table with the following simple schema: \texttt{synthetic(time:timestamp, time_series: float)}, where the column \texttt{time} is the primary key, indexed by the default PostgreSQL DB index (B-tree data structure); 
we note that indexing the time column is a standard practice in time series DBs.

\vspace{1mm}
{\em Throughput.}
As before, we evaluate the time taken to create a prediction model in tspDB as we vary the number of rows inserted from {$10^4$ to $10^8$} and compare it against the time needed to insert the same rows into the aforementioned indexed table.
We find that the time required to train tspDB's  prediction model  is {$1.13$-$2.17$x} faster than PostgreSQL insert time as we vary the dataset size. 
In absolute terms, the time to train tspDB's  prediction model is {$2.62$-$7.78$} microseconds per record.
See Figure \ref{fig:thpt}.
Since the time required to create the prediction model in tspDB is comparable to PostgreSQL's bulk insert time, and the two operations are independent and so can be effectively run in parallel, integrating tspDB with PostgreSQL will pleasingly not bottleneck PostgreSQL's throughput as new data points are inserted.

\vspace{1mm}
{ \em Query Latency.}
As before, we compare the latency of a standard SELECT point query in PostgreSQL against the a PREDICT query in tspDB. 
We evaluate the latency for both imputation and forecasting predictions with and without UQ. 
As shown in Figure \ref{fig:latency}, we find that imputation PREDICT queries are  {$1.77$-$2.56$x} slower relative to SELECT queries as we vary the table size, and forecasting PREDICT queries are {$3.87$-$6.39$x} slower.
PREDICT queries with UQ are {$3.04$-$4.60$x} and {$6.99$-$12.80$x} slower for imputation and forecasting queries respectively.
In absolute terms, this amounts to $0.41$-$0.61$ milliseconds for SELECT queries, $0.94$-$1.23$ milliseconds for imputation queries ($1.72$-$2.11$ milliseconds with UQ), and  $1.73$-$3.04$ milliseconds for  forecasting queries ($3.02$-$5.85$ milliseconds with UQ).

\iftoggle{FIGURES}{
\begin{figure}[!htb]
\vspace{-2mm} 
\floatconts
{fig:var}
{\caption{(a) Training tspDB prediction model is $1.13$-$2.17$x faster than PostgreSQL's bulk insert. (b) Predictions using tspDB  are only $1.77$x to $6.39$x slower than standard SELECT queries.}}
{
	\subfigure[][t]{\label{fig:thpt}
		\includegraphics[width=0.38\linewidth]{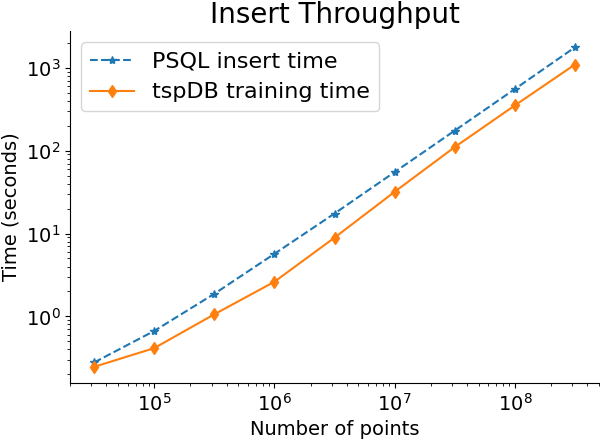}}
	\quad 
	\subfigure[][t]{\label{fig:latency}
	\includegraphics[width=0.38\linewidth]{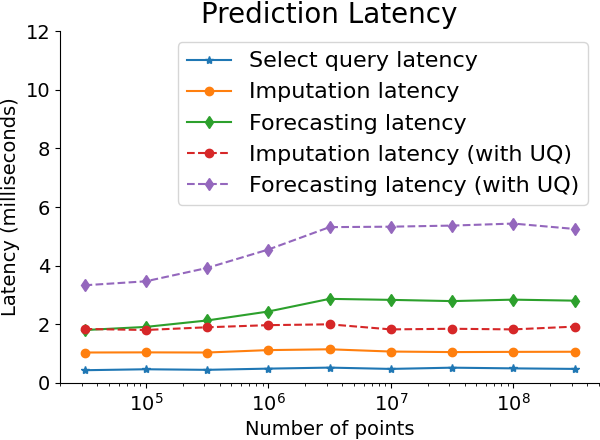}}
	\vspace{-5mm} 
}
\vspace{-7mm} 
\end{figure}
}


\vspace{2mm}
\section{Statistical and Computational Tradeoffs in tspDB} \label{sec:appendix_stat_experiments}
We detail two additional set of experiments we conduct to study the statistical and computational tradeoffs in our implementation of tspDB.
In Section \ref{sec:multi_column}, we explore this tradeoff between multivariate and univariate time series prediction models -- we find using multiple time series can greatly improve accuracy with a moderate slow down in training time and query latency. 
However, beyond a certain point, the gain in statistical accuracy from adding more time series ($ \sim N > 50$) becomes minimal and can add unnecessary overhead.
In Section \ref{sec:appendix_comp_experiments}, we study how the three key parameters in the scalable implantation of the prediction algorithm in tspDB, $T_0, T', \gamma$ (see Section \ref{sec:inc_algorithm} for their definition), affect tspDB's statistical and computational performance.
We select the default parameters in tspDB based on these experiments.

\subsection{Multivariate vs. Univariate Prediction Models in tspDB} \label{sec:multi_column}
In our proposed algorithm, a prediction model is trained using data from a collection of time series. 
If this collection of time series are \textit{carefully} selected (i.e., have high ``correlation"), this should lead to better statistical accuracy (see Section \ref{sec:justification} for a justification).
However, the drawback is that  training on multiple time series leads to higher model training time and prediction query latency. 
We evaluate this tradeoff in tspDB's performance across the three metrics of interest.

\hspace{4mm} Specifically, we use  the synthetic I and electricity datasets (detailed in Appendix \ref{sec:data_desc}).
We evaluate tspDB's performance as we vary the number of time series, $N$, used in training the prediction model.
{$N \in \{1, 10, 40, 80, 160, 400\}$} and {$N \in \{1,10,40,80, 150, 370\}$} for the synthetic and electricity datasets, respectively.
We report performance across the three metrics relative to the univariate case.
To reduce variance due to randomness in which time series are selected for the statistical accuracy evaluation, we average results over $50$ runs.

\begin{table}[]
\centering
\caption{How the three metrics change as we train using more time series ($N$) relative to the univariate case.}
\label{table:multi_vs_uni}
\fontsize{8.9pt}{9.9pt}\selectfont
\tabcolsep=0.11cm
\begin{tabular}{@{}lllllllllllll@{}}
\toprule
              Dataset       & \multicolumn{6}{c}{Synthetic I }    & \multicolumn{6}{c}{Electricity }       \\ \midrule
N                  & 1    & 10   & 40   & 80    & 160   & 400   & 1    & 10   & 40    & 80    & 150    & 370    \\
Forecasting Accuracy & 1.00 & 0.62 & 0.11 & 0.11  & 0.08  & 0.06  & 1.00 & 0.49 & 0.44  & 0.45  & 0.45   & 0.46   \\
Training Time        & 1.00 & 6.80 & 16.2 & 30.28 & 55.24 & 250.2 & 1.00 & 6.53 & 20.84 & 39.63 & 127.13 & 369.06 \\
Forecasting Latency  & 1.00 & 1.00 & 1.05 & 1.15  & 1.28  & 1.60  & 1.00 & 1.28 & 1.35  & 1.60  & 1.931  & 2.62   \\ \bottomrule
\end{tabular}
\end{table}

\vspace{2mm}

\textbf{Results.} 
Table \ref{table:multi_vs_uni} shows how the three metrics change as we train using more time series. 
The training time increases almost linearly, as we train on more time series in the synthetic ($6.80$x-$250.2$x) and electricity ($6.53$x-$369.06$x) datasets.
The prediction query latency increases moderately by {$1.001$x-$1.60$x} and {$1.28$x-$2.62$x} for the synthetic and electricity datasets respectively.  
In contrast, the forecasting error decreases by {$1.61$x-$17.59$x} and {$2.04$x-$2.27$x} for the synthetic and electricity datasets respectively.
In summary, we find that using multiple time series to train the prediction model can greatly improve accuracy with a moderate slow down in the training time and query latency.
However, beyond a certain point, the gain in statistical accuracy from adding more time series ($\sim N > 50$) becomes minimal and thus can add unnecessary overhead.
%

\subsection{Hyper-parameter Tuning - $T_0, T', \gamma$} \label{sec:appendix_comp_experiments}

We discuss how the default parameters for $T_0, T', \gamma$ are selected in tspDB -- 
we quantify the effect of each of these parameter on three metrics of interest: training time, prediction latency, and statistical accuracy. 
We find $T_0$ has minimal effect on these three metrics, and for simplicity of exposition, we simply choose it to be $T_0 = 100$ for all experiments. 

\vspace{2mm}

\textbf{Setup.} 
In this experiment, we use a synthetic time series generated from a sum of harmonics  with {$5 \times 10^7$} data points.
We initially train the prediction model with all  but the last {$10^6$} entries and then incrementally add data over a {$1000$} batches, with a {$1000$} data points in each batch;
note we do so to study the effect of $\gamma$ on training time, whose effect is limited to such data insertion patterns. 
We vary the parameter {$T' \in [10^4,10^7]$, and $\gamma \in \{0.01, 0.2,0.5,0.7,1.0\}$}. 

\vspace{2mm}

\textbf{Results.} 
Figure \ref{fig:parameter_sweep} shows tspDB's performance on the three metrics of interest as we vary the parameters. 
The x-axis reflects the time needed to train the prediction model relative to the time needed to insert the same data into a PostgreSQL table. 
The y-axis reflects the forecast query latency relative to a PostgreSQL SELECT query. 
The color of each dot corresponds to the forecasting accuracy in normalized RMSE.
We see there is a clear trade-off between query latency and accuracy.
Specifically, using small values of $T' \in [10^4, 10^5]$  yields faster but less accurate predictions  ({$3$x-$4$x} slower than SELECT queries). 
Further, it has a negative effect on training time (it is {$2$x-$6$x} the insert time). 
On the other hand, using high values of {$T' \in [5\times 10^6, 5 \times 10^7]$} yields more accurate but slower predictions ({$6$x-$7$x} slower than SELECT queries), with relatively low training time (train time is {$70\%$} of  insert time).
We choose the default setting of $T'$ to be {$2.5 \times 10^6$}, gives accurate and relatively fast predictions ({$5.1$x} slower than SELECT queries) and relatively low training time (training time is {$1.25$x} quicker compared to PostgreSQL insert time).

\hspace{4mm} We find $\gamma$'s effect on the accuracy to be minimal but has significant effect on training time for certain range of parameter choice. 
For example choosing {$\gamma = 0.01$ slows training time by $\sim 50\%$ compared to $\gamma = 1.0$}. 
Training time does not vary for {$\gamma \in \{0.5,0.7,1.0\}$, with choice of $\gamma = 0.5$} performing best.
Hence we choose {$\gamma = 0.5$.} to be the default setting. 
Note $\gamma$ has no effect on predictions  latency. 

\iftoggle{FIGURES}{
\begin{figure}[!htb]
	\centering
	\includegraphics[width=0.6\linewidth]{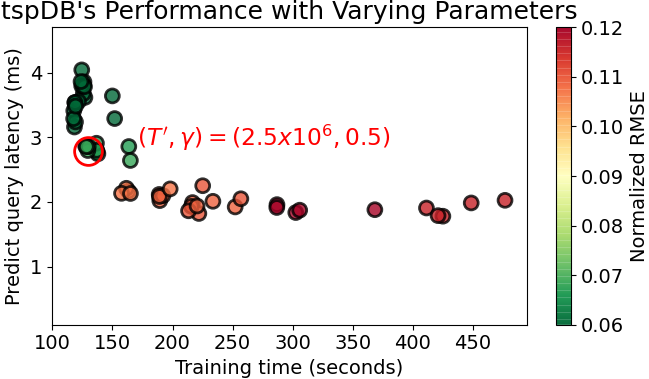}
	\caption{$T'$ and $\gamma$ can significantly affect the three metrics of interest. We choose {$T' = 2.5 \time 10^6$ and $\gamma = 0.5$} which gives the best tradeoff across all three metrics.}
	\label{fig:parameter_sweep}
\end{figure}
}

\vspace{-3mm}
\section{Conclusion}
\vspace{-2mm}
In this paper, we build an open-source, real-time time series prediction system tspDB's, which achieves state-of-the-art statistical and computational performance with respect to widely used time series algorithms/libraries, including deep-learning based methods such as LSTMs and DeepAR.
{tspDB's excellent statistical performance is likely due to real-world time series data following the time series model we introduced in Section \ref{sec:justification}.
Given that this one time series model captures a broad class of time series dynamics, and the matrix factorization algorithm we utilize is well-motivated under it, it allows us to focus our efforts on building a scalable, incremental implementation of said algorithm.
%
This along with the surprising simplicity of the algorithm used in tspDB, simply doing singular value thresholding and linear regression---compared to more complicated non-linear models such as LSTM and DeepAR---likely explains the large gains in computational performance in tspDB.
}

\hspace{4mm} 
We recall that an important design choice we make is to directly integrate tspDB on top of PostgreSQL and abstract away the ML workflow from a user, thus striving for a single interface to answer both standard DB queries and predictive queries.
We hope this serves as a benchmark for other such integrated prediction systems in the growing and exciting line of work in AutoML and Systems for ML.

\bibliography{vldb_sample}  

\newpage 

\appendix

\section{Appendix Organization } \label{sec:appendix_org}
The appendix consists of two main sections.

\vspace{2mm}

\textbf{Related Work.}
In Appendix \ref{sec:related_work}, we discuss prior related work both with respect to time series analysis and the interplay between ML and DBs.

\vspace{2mm}
\textbf{Experimental Setup Details.}
In Appendix \ref{sec:appendix_exp_setup}, we detail the experimental setup used, in particular: the datasets used; the computational infrastructure and DB configurations used, the hyper-parameters used for tspDB and the other time series algorithms compared against; and the statistical accuracy metrics used.
%

\section{Related Work}\label{sec:related_work}

\textbf{Predictive Modeling and DBs.} 
As stated earlier, our work is in line with exciting recent efforts, especially in industry, exploring the potential of DB management systems to support ML algorithms~\cite{cxoracle, pymssql, ibmdb, RODM, RevoScaleR, idmdbR, madlib, mldb.ai}.
In a complementary vein, there has been a line of work to automate the process of choosing an optimal ML workload for a given task \cite{Tupaq, alpine}.
There has also been a line of work to exploit DBs to make certain key computationally expensive steps that are pervasive in training ML algorithms, much more efficient~\cite{blinkML, sparkdec, madskills, sciDB, denormalizedlinear}.
Further, there have been attempts to provide a declarative SQL-like language for prediction \cite{PQL, madlib, thecase}.
tspDB is very much in line with these efforts.
However, crucially rather than facilitating the use of a variety of ML algorithms or sub-routines, we take the stance of abstracting the entire ML algorithm from the user, and providing a single interface to answer both standard DB queries and predictive queries.
This is in line with \cite{Tupaq, alpine}.
Our focus is specifically on prediction with time series data, an important part of ML with crucial applications, that has not yet been directly tackled.
To justify abstracting the ML algorithm from the user, we comprehensively test both the statistical and computational performance of tspDB compared to state-of-the-art alternatives, and show promising results.
We note that just the first prediction task, imputation, by itself has received considerable attention in the systems for ML community. 
For example, in \cite{impsurvey}, they conduct a thorough empirical evaluation of twelve prominent imputation techniques and in \cite{imputDB, RecovDB, ERACER, BayesDBimp}, they explore systems which directly integrate imputation of missing values as an in-built functionality.
In addition to imputation, tspDB also adds forecasting predictive functionality. And, importantly it does both with uncertainty quantification.

\vspace{2mm}

\textbf{Modern Multivariate Time Series Algorithms.}
Given the ubiquity of multivariate time series analysis, we cannot possibly do justice to the entire literature. 
Hence we focus on a few techniques, most relevant to compare against empirically.
Recently, with the advent of deep learning, neural network (NN) based approaches have been the most popular, and empirically effective. Some industry standard NN methods include LSTMs and DeepAR (for an industry leading NN library for time series analysis, see \cite{DeepAR}).

\vspace{2mm}

\textbf{Matrix Factorization Based Time Series Algorithms.}
The line of time series analysis most closely associated with the algorithm we propose is where a matrix is constructed from a time series, and some form of matrix factorization is done (see \cite{matfac2, TRMF}). 
A particularly relevant algorithm is Singular Spectrum Analysis (SSA) (see \cite{SSA}).
The main steps of SSA are: 
Step 1 - create a Hankel matrix from univariate time series data; 
Step 2 - do a Singular Value Decomposition (SVD) of it; 
Step 3 - group the singular values based on user belief of the model that generated the process; 
Step 4 - perform diagonal averaging to “Hankelize” the grouped rank-1 matrices outputted from the SVD to create a collection of time series; 
Step 5 - learn a linear model for each “Hankelized” time series for the purpose of forecasting. 
A generalization of SSA to multivariate time series data is called multivariate SSA (mSSA); 
here the only change is in Step 1, where a Hankel matrix is constructed for each time series, and the various matrices are concatenated column-wise to create a ``stacked'' Hankel \cite{mSSA1, mSSA2}.

\vspace{1mm}

In \cite{sigmetrics}, the authors show that a simple variant of SSA, where close to no hyper-parameter tuning is required, is both theoretically and empirically effective.
They show that by using the Page matrix instead of the Hankel (see Section \ref{sec:overview} for the definition), doing a single SVD on it, and subsequently learning a single linear forecaster suffices. 
This is in contrast to the user having to specify how the singular values need to be grouped in the original SSA method.

\vspace{1mm}

The multivariate time series algorithm we introduce is a variant of mSSA and a generalization of the univariate time series algorithm proposed in \cite{sigmetrics}  -- 
inspired by mSSA, the core data structure for multivariate time series data we analyze is the ``stacked'' Page matrix (the matrix induced by column-wise concatenation of the Page matrices constructed for each time series);
in line with \cite{sigmetrics}, we show that even for multivariate time series data, surprisingly one only needs to do a single SVD and learn a single linear forecaster to produce effective predictions. 
See Section \ref{sec:justification} for a theoretical justification for why our proposed variant of mSSA works, something which importantly has been absent from the literature. 

\vspace{1mm}

Another popular method in the literature is TRMF (see \cite{TRMF}) and one we directly compare against due to its popularity and empirical effectiveness for both imputation and forecasting.

\vspace{2mm}

\textbf{Time-Varying Variance Estimation.}
The time-varying variance is a key input parameter in many sequential prediction algorithms themselves. 
For example in control systems, the widely used Kalman Filter uses an estimate of the per step variance for both filtering and smoothing. 
Similarly in finance, the time-varying variance of each financial instrument in a portfolio is necessary for risk-adjustment. 
The key challenge in estimating the variance of a time series (which itself might very well be time-varying) is that unlike the actual time series itself, we do not get to directly observe the variance. 
Despite the vast time series literature, existing algorithms to estimate time-varying variance are mostly heuristics and/or make restrictive, parametric assumptions about how the variance (and the underlying mean) evolves; for example “Auto Regressive Conditional Heteroskedasticity” (ARCH) and  “Generalized Auto Regressive Conditional Heteroskedasticity” (GARCH) models. 
See \cite{Bollerslev}.



 
\section{Experimental Setup} \label{sec:appendix_exp_setup}
 \vspace{2mm}

\subsection{Datasets Description}\label{sec:data_desc}
Throughout the experiments, we use three real-world datasets that are standard benchmarks in time series analysis as well as three synthetic datasets. 
In this section, we describe these datasets in details. 

 \vspace{2mm}

{\bf Electricity Dataset. }
Obtained from the UCI repository, it records 15-minutes electricity loads of 370 households (\cite{ucielec}). 
We follow the preprocessing done in \cite{TRMF, Amazon, DeepAR}: data is aggregated into hourly intervals; the first 25968 time-points are used for training; and day-ahead forecasts for the next seven days (i.e. 24-step ahead for 7 windows) are made.

\vspace{1mm}
{\bf  Traffic Dataset. }
Obtained from the UCI repository, it records occupancy rate of traffic lanes in San Francisco (\cite{ucitraffic}). 
We follow the preprocessing done in \cite{TRMF, Amazon, DeepAR}: data is aggregated into hourly intervals; first 10392 time-points are used for training; day-ahead forecasts for the next seven days (i.e. 24-step ahead for 7 windows) are made.

\vspace{1mm}
{\bf Financial Dataset.}
Obtained from Wharton Research Data Services (WRDS), it records average daily stocks prices of 839 companies from October 2004 -- November 2019 \cite{WRDS}. 
To limit number of time series for ease of experimentation, we remove stocks with average prices below 30\$ across the available period, and those with null values.
This preprocessing gives 839 time series (i.e. stock prices) each with 3993 daily readings.  
We use the first 3813 time points for training. 
For forecasting, we forecast 180 days ahead one day at a time (i.e. 1-step ahead forecast for 180 windows).
We do so as this is a standard goal in finance.

\vspace{1mm}
{\bf  Synthetic Dataset I (Mean Estimation). }
We generate observations of $n \times m$ time series over $T$ observations $X \in \mathbb{R}^{ n \times m \times T} $ by first randomly generating the two vectors $U \in \mathbb{R}^{n} =[u_1,\dots, u_n]$ and $V \in \mathbb{R}^{m} =[v_1, \dots, v_n]$. 
Then, we generate r mixtures of harmonics where each mixture $g_k(t), k \in [r],$ is generated as: $g_k(t) = \sum_{h=1}^4 \alpha_h \cos(\omega_h t/T) $
where the parameters are selected randomly such that  $\alpha_h \in [-1,10]$ and $\omega_h \in [1,1000]$.
Each observation is constructed as follows:
$
X_{i,j}(t) = \sum_{k=1}^r u_i  v_j g_k(t)
$
where $r$ is the tensor rank, $i \in [n]$, $ j\in [m]$. 
In our experiment, we select $n = 20$, $m = 20$, $T = 15000$, and $r =4$.  
This gives us 400 time series each with 15000 observations.
In the forecasting experiments, we use the first 14000 points for training.
The goal here is to do 10-step ahead forecasts for the final 1000 points.

\vspace{1mm}
{\bf  Synthetic Dataset II (Variance Estimation). }
With $k \in \{1,2,3,4\}$, we generate three different sets of time series as follows:
(i) four harmonics $g^{har}_k(t) = \sum_{i=1}^4 \alpha_i \cos(\omega_i t/T) $ where  $\alpha_i \in [-1,10]$ and $\omega_i \in [1,1000]$;
(ii) four AR processes $g^{AR}_k(t) = \sum_{i=1}^3\phi_i  g^{AR}_k(t-i) + \epsilon(t) $ where  $\epsilon(t) \sim \mathcal{N}(0,0.1) $, $\phi_i \in [0.1,0.4]$;
(iii) four trends: $g_k^{trend}(t) =  \eta t$ where  $\eta \in [10^{-4},10^{-3}]$.
Then we generate a tensor by sampling two random vectors  $U \in \mathbb{R}^{20} =[u_1,\dots, u_{20}]$ and $V \in \mathbb{R}^{20} =[v_1, \dots, v_{20}]$.
Next we generate three $20\times 20 \times 1500$ tensors as follow: 
 (i) A mixture of harmonics:  $F^{1}_{i,j}(t) = \sum_{k=1}^4 u_i  v_j g^{har}_k(t)$;  
 (ii) A mixture of harmonics + trend: $F^2_{i,j}(t) =  \sum_{k=1}^4 u_i  v_j (g^{har}_k(t)+ g^{trend}_k(t))$; 
 (iii) A mixture of harmonics + trend+ AR: $F^3_{i,j}(t) =   \sum_{k=1}^4 u_i  v_j ( g^{trend}_k(t)+ g^{har}_k(t) + g^{AR}_k(t))$.
Here $i, j \in [1, \dots, 20]$.
For each tensor, we normalize its values to be between 0 and 1 and  use three observations models to generate three tensors from each:
(i) Gaussian: $ X^q_{ij}(t)  \sim \mathcal{N}(F^{1}_{i,j}(t), F^{q}_{i,j}(t))$;
(ii) Bernoulli: $X^q_{ij}(t) \sim  \text{Bernoulli}(F^{q}_{i,j}(t))$; %
(iii) Poisson:  $X^q_{ij}(t) \sim  \text{Pois}(F^{q}_{i,j}(t))$. 
Where  $q \in \{1,2,3\}$. 
Note, this give us a total of nine observation tensors for our variance experiments.
For the forecasting experiments, we use the first 14000 points for training and the last 1000 time steps for testing.

\vspace{1mm}
%

{\bf  Synthetic Dataset III (Robustness to Observation Models). }
We generate a $f$ as a sum of harmonics; $f(t) = \sum_{i=1}^4 \alpha_i \cos(\omega_i t/T) $, where $t \in [T]$, $T = 100000$, $\alpha_i \in [-1.5,1.5]$ and $\omega_i \in [1,100]$ (randomly selected).
We normalize $f(t)$ to be between 0 and 1, and then generate three different observation models:
(i) $X_{Gauss}(t) = f(t) + \epsilon(t)$, where $\epsilon(t) \sim \mathcal{N}(0, 0.5)$ (float);
(ii) $X_{Bernoulli}(t) \sim  \text{Bernoulli}(f(t)) $, i.e. each observation  at time t follows a Bernoulli distribution with mean $f(t)$ (boolean);
(iii) $X_{Pois}(t) \sim  \text{Pois}(f(t)) $, i.e. each observation  at time t follows a poisson distribution with mean $\lambda = f(t)$ (integer). %
The goal is to recover (i.e. impute) and forecast $f(t)$ for each set of observations. 
Note that we use the first $96000$ observations for training, and the last $4000$ points for testing. 

\vspace{1mm}

\subsection{Machine and DB Configuration.}  \label{sec:system_config}
As we said earlier, we use an Intel Xeon  E5-2683 machine with 16 cores, 132 GB of RAM, and an SSD storage in all experiments. 
In Table \ref{DBSpec},  we detail the relevant settings used for PostgreSQL 12.1. 

\begin{table}[h]
	\caption{The used configurations for PostgreSQL 12.1.}
	\label{DBSpec}
	\centering
	\begin{tabular}{@{}ll@{}}
		\toprule
		Parameter & Value \\ \midrule
		Shared\_buffers & 30GB \\
		maintenance\_work\_mem & 2GB  \\
		effective\_cache\_size & 80GB  \\
		default\_transaction\_isolation & `read committed' \\
		wal\_buffers & 16MB \\
		max\_parallel\_workers & 16  \\
		max\_worker\_processes & 16 \\
		work\_mem & 54MB \\ \bottomrule
	\end{tabular}
\end{table}

\subsection{tspDB Hyperparameters}  \label{sec:hyperparameters_selection}
All experimental results are produced using the default settings in the open-source implementation of tspDB.
In particular, the hyperparameters of tspDB are selected as follow:

\vspace{1mm}

\textbf{L.} 
We choose L using guidance from the analysis in \cite{sigmetrics}, which requires $L\leq\bar{P}$. 
Specifically in tspDB, $L = \bar{P} / 10$.

\vspace{1mm}

\textbf{Number of Singular Values ($k, k_1,k_2$).}  
This choices is done in a data-driven manner. 
Specifically, we follow the optimal procedure proposed in \cite{Donoho} for choosing a threshold for HSVT.
The procedure depends on the median of the singular values and the shape of the Page matrix.
For more details, refer to  \cite{Donoho}.

\vspace{1mm}

\textbf{``Meta-Algorithm'' Parameters ($T_0, T', \gamma$).}  
The default parameters we choose are $T_0 = 100, \ T' = 2.5\times10^6, \ \gamma = 0.5$.
For details on how these parameters are chosen, refer to Section \ref{sec:appendix_comp_experiments}.

\subsection{Benchmarking Algorithms Parameters and Settings}\label{sec:exp_appendix_methods}
In this section, we  describe the hyper-parameters/implementation used for each method we compare with.

\vspace{1mm}
{\bf DeepAR. } We use the default parameters of ``DeepAREstimator'' algorithm provided by the GluonTS package.
In variance forecasting, we compare with DeepAR as it natively had functionality to produce prediction intervals.
It uses a Monte Carlo approach that samples the estimated posterior to produce multiple samples of the time series trajectories.
We take the variance of these samples as an estimate of the forecasted variance.

\vspace{1mm}
{\bf LSTM.} Across all datasets, we use a LSTM network with three hidden layers each, with 45 neurons per layer, as is done in~\cite{Amazon}. We use the Keras implementation of LSTM~\cite{Keras}.
We train the network with a batch size of 128 for 100 epochs. 
%

\vspace{1mm}
{\bf Prophet. } We used Prophet's Python library with the default parameters selected~\citep{Prophet}. 

\vspace{1mm}
{\bf TRMF. } 
We use the implementation provided by the authors (\cite{TRMF}) in  Github. 
We use $k = 60$ for the electricity, $k = 40$ for the  traffic, $k =20$ for the financial  and synthetic data and variance experiments (selected via cross-validation);
$k$ represents the chosen rank of the $T\times N$ time series matrix. 
Another hyper-parameter, the lag index is set to include the last day and same weekday in the last week for the traffic and electricity data. 
For the financial and synthetic data, the lag index include the last 20 points.
To perform variance imputation, we adapt TRMF similar to the extension used in tspDB as follows:
\begin{enumerate}
\item Impute the time series $X_n(t), X^2_n(t)$ to produce estimates $ \hat{f_n}^{\text{TRMF}}(t), \widehat{f_n^2 + \sigma_n^2}^{\text{TRMF}}(t)$ respectively;  
\item Output $\hat{\sigma_n^2}^{\text{TRMF}} = \widehat{f_n^2 + \sigma_n^2}^{\text{TRMF}}(t) -\hat{f}^{\text{TRMF}}_n(t)^2 $.
\end{enumerate}

\subsection{Accuracy Score Metrics}\label{sec:score_metric}
In this section, we provide detailed description for the accuracy metrics used throughout the experiments. 

 \vspace{2mm}
 
{\bf Normalized Root Mean Squared Error (NRMSE). }  
Throughout all experiments, NRMSE is used as an accuracy metric.
In particular, we rescale each time series to have a zero mean and unit variance before calculating the RMSE.
This is done to avoid over-weighting the error in time series with larger magnitudes or greater variance.

 \vspace{2mm}
 
{\bf  Weighted Borda Count (WBC). } 
In Table \ref{table:full} and Table \ref{table:stat}, we use WBC to report the algorithms' performance across different experiments. 
This score is inspired from Borda count --  a commonly used method to rank a list of candidates.
Specifically, rather than choosing a single candidate, voters rank all candidates in order of preference. 
Then, each candidate is given a number of points that corresponds to the number of candidates ranked lower.

\hspace{4mm} We use a weighted version of the standard Borda Count, where we rank algorithms (i.e. candidates) based on pairwise comparisons of the normalized root mean squared error (NRMSE) across experiments (i.e. voters). 
We use this metric as it better captures the relative performance across all experiments and methods, unlike a simple statistic such as the mean or the median. 
Particularly, let $\mathcal{A}$ represent the set of algorithms used in the experiments. 
For example, in the mean forecasting experiments, $\mathcal{A} = \{\text{tspDB}, \text{ DeepAR}, \text{ LSTM}, \text{ Prophet}, \text{ TRMF}\}$. 
Let $\mathcal{X}$ represent the set of all experiments across different  noise levels ($\sigma$), and fraction of missing values ($p$).
Let $e(a,x)$ represent the NRMSE of algorithm $a \in \mathcal{A}$ in experiment $x \in \mathcal{X}$. 
Then the Weighted Borda Count of algorithm $a \in \mathcal{A}$ is:
\[
\text{WBC}(a)= \frac{1}{|\mathcal{A}| - 1} \sum_{a' \in \mathcal{A}: a' \ne a } \left(  \frac{1}{|\mathcal{X}|}  \sum_{x \in \mathcal{X}} \frac{e(a',x)}{e(a,x)+e(a',x)} \right)
\]
WBC ranges  between $0$ and $1$ and captures the relative performance against alternative algorithms. 
A score of $0.5$ indicates identical performance to the average of all other algorithms, a score of $0 / 1$ indicates it performs ``infinitely" worse / better, respectively.

\end{document}